\documentclass[preprint,review,12pt]{elsarticle}%
\usepackage{amsmath,amssymb,amsthm,epsfig,graphicx,color}
\usepackage{booktabs}
\usepackage{enumitem}
\usepackage{amsmath}
\usepackage{amsfonts}
\usepackage{amssymb}
\usepackage{graphicx}%
\usepackage[running,switch,pagewise]{lineno}
\newcommand{\bu}{{\bf u}}

\newcommand{\BEA}{\begin{eqnarray}}
\newcommand{\EEA}{\end{eqnarray}}
\newcommand{\comment}[1]{}
\newtheorem {theorem} {Theorem}
\newtheorem {proposition} [theorem] {Proposition}

\newtheorem {definition} {Definition}

\newtheorem {example} {Example}

\journal{Journal of Computational Physics}
\begin{document}

\begin{frontmatter}
\title{An Algebraic Method for Constructing Stable and Consistent Autoregressive Filters}
\author{John Harlim}
\ead{jharlim@psu.edu}
\address{Department of Mathematics and Department of Meteorology, the Pennsylvania State University, University Park, PA 16802}
\author{Hoon Hong}
\ead{hong@ncsu.edu}
\author{Jacob L. Robbins\corref{cor}}
\cortext[cor]{Corresponding author.}
\ead{jlrobbi3@ncsu.edu}
\address{Department of Mathematics, North Carolina State University, Raleigh, NC 27695}
\begin{abstract}
In this paper, we introduce an algebraic method to construct stable and consistent univariate autoregressive (AR) models of low order for filtering and predicting nonlinear turbulent signals with memory depth. By stable, we refer to the classical stability condition for the AR model. By consistent, we refer to the classical consistency constraints of Adams-Bashforth methods of order-two. One attractive feature of this algebraic method is that the model parameters can be obtained without directly knowing any training data set as opposed to many standard, regression-based parameterization methods. It takes only long-time average statistics as inputs. The proposed method provides a discretization time step interval which guarantees the existence of stable and consistent AR model and simultaneously produces the parameters for the AR models.
In our numerical examples with two chaotic time series with different characteristics of decaying time scales, we find that the proposed AR models produce significantly more accurate short-term predictive skill and comparable filtering skill relative to the linear regression-based AR models. These encouraging results are robust across wide ranges of discretization times, observation times, and observation noise variances. Finally, we also find that the proposed model produces an improved short-time prediction relative to the linear regression-based AR-models in forecasting a data set that characterizes the variability of the Madden-Julian Oscillation, a dominant tropical atmospheric wave pattern.
\end{abstract}
\end{frontmatter}

\section{Introduction}
Filtering or data assimilation is a numerical scheme for finding the best statistical estimate of the true signals from noisy partial observations. In the past two decades, many practical Bayesian filtering approaches \cite{lorenc:86,evensen:94,anderson:01,hunt:07} were developed with great successes in real applications such as weather prediction and assimilating high dimensional dynamical systems. Despite these successful efforts, there is still a long-standing issue when the filter model is imperfect due to unresolved scales, unknown boundary conditions, incomplete understanding of physics, etc. Various practical methods have been proposed to mitigate model errors. In the mildest form, model error arises from misspecification of parameters. In this case, one can, for example, estimate the parameters using a state-augmentation approach in the Kalman filtering algorithm \cite{friedland:69,friedland:82,dds:98}. A more difficult type of model error is when the dynamics (parametric form) of the governing equations of the underlying truth is (partially) unknown. Recently, reduced stochastic models were proposed for filtering with model errors in complex turbulent multiscale systems \cite{mg:12,bcm:13,mh:13,gh:13,hmm:14,bh:14}. While most of these methods are very successful, they require partial knowledge about the underlying processes and they are mostly relevant in many applications with known, but coarsely resolved, governing equations such as climate and weather prediction problems.

In this paper, we consider a challenging situation in which we completely have no access to the underlying dynamics. The only available information are some equilibrium statistical quantities of the underlying systems such as energy and correlation time. The Mean Stochastic Model (MSM) \cite{hm:08b,mgy:10}, an AR model of order-1, 
was the first model designed to mitigate this situation. The resulting model produces accurate filtering of nonlinear signals in the fully turbulent regime with short decaying time (or short memory). Indeed, filtering with MSM was shown to be optimal in the linear and Gaussian settings \cite{bh:14}. For weakly chaotic nonlinear time series with long decaying time scale (memory), higher-order linear autoregressive (AR-p) filters were shown to be more accurate \cite{kh:12b}. In linear and Gaussian setting, the univariate AR-p filter is ``optimal" when the parameters in the model are chosen to satisfy the classical stability criterion and the Adams-Bashforth's consistency conditions of order-2 \cite{bh:13b}. While the stable and consistent AR model was shown to improve the filtering skill of nonlinear systems in some cases, there are examples in which such a model is not even attainable for some choices of $p$ and sampling time $\delta t$ \cite{bh:13b}.  The central contribution in this paper is on an algebraic-based method for constructing (or parameterizing) stable and consistent one-dimensional (univariate) AR models of low order. We require our method to provide an upper bound for discretization time step, $\delta t$, to guarantee the existence of the stable and consistent AR models. Moreover, we require the new method to only take equilibrium statistics as inputs rather than a training data set. We will provide numerical tests on synthetic examples, two Fourier modes of the Lorenz-96 model \cite{lorenz:96} with different characteristic of decaying time scales (or memory depth), and on a data set from real world problem, the Realtime Multivariate MJO (RMM) index \cite{wh:04} that characterizes the dominant wave pattern in the tropical atmosphere, the Madden-Julian Oscillation \cite{zhang:05}.  

This paper is organized as follows: We briefly review the stable and consistent AR model that guarantees accuracy of the filtered solutions in the linear setting when Kalman filter is used in Section~2. Therein, we point out the main issues when this linear result is applied on nonlinear problems. We also include a standard parameterization method for AR models to keep the paper self-contained. In Section~3, we discuss the proposed method for finding stable and consistent AR models in detail using some standard tools from algebraic geometry. In this section, we also provide a numerical example to help illustrate the algorithm. To keep the paper self-contained, we list the relevant definitions and results in algebraic geometry in the Appendix. In Section~4, we will numerically compare the filtered posterior and prior estimates of the newly developed AR models with estimates from the classical, regression-based, AR models. We will summarize the paper in Section~5.

\section{Stable and consistent linear autoregressive Kalman filter}

The goal of this section is to briefly review the main theoretical result in \cite{bh:13b} which states that: the approximate mean estimates from Kalman
filter solutions with a stable and consistent AR model are as accurate as the optimal filter estimates in the linear setting. Then, we will discuss practical issues in the nonlinear setting and motivate the importance of new parameterization algorithms that satisfy the theoretical constraints in \cite{bh:13b} even when the training data set is not directly available to us.

To keep this paper self-contained, we briefly review some necessary background
materials, including: the AR model, a standard linear regression method to
parameterize the AR model, the consistency constraints of Adams-Bashforth methods for numerical approximation of ODE, and the Kalman filtering algorithm with AR model.

\subsection{Linear AR model}

Consider a linear discrete-time autoregressive (AR) model of order-$p$ for $u \in\mathbb{C}$. In vector form, the linear AR model of
order-$p$ is given as follows,
\begin{align}
\mathbf{u}_{m+1} = F_{p} \mathbf{u}_{m} + \mathbf{f} +\mathbf{e}_{m+1},
\label{AR}%
\end{align}
where $\mathbf{u}_{m} = (u_{m-p+1},\ldots,u_{m-1},u_{m})^{\top}$ and
$\mathbf{f} = (0,\ldots,0,f)^{\top}$ are $p$-dimensional vectors of the
temporally augmented variables and a constant forcing term, respectively. In
\eqref{AR}, $m$ denotes discrete time index with an interval, $\delta t=t_{m+1}-t_{m}$, 
which is typically chosen to be the sampling time interval of the given training dataset. 
For an AR model of order-$p$, the  deterministic operator $F_{p}$ is a
$p\times p$ matrix with the following entries,
\begin{align}
F_{p} =
\begin{pmatrix}
0 & 1 & 0 & \cdots & 0 & 0\\
0 & 0 & 1 & \cdots & 0 & 0\\
\vdots & \vdots &  & \ddots & \vdots & \vdots\\
0 & 0 & 0 & \cdots & 0 & 1\\
a_{1} & a_{2} & a_{3} & \cdots & a_{p-1} & 1+a_{p}%
\end{pmatrix}
. \label{F}%
\end{align}
The system noises, $\mathbf{e}_{m} = (0,\ldots,0,\eta_{m})^{\top}$, are
represented by $p$-dimensional vectors, where $\eta_{m}$ are i.i.d.~Gaussian
noises with mean zero and variance $Q$.

Recall that the characteristic polynomial of matrix $F_{p}$ is given by
\begin{align}
\Pi(x) = \sum_{j=1}^{p}{a_{j}}x^{j-1}+x^{p-1}-x^{p}, \label{Pi}%
\end{align}
and it is well known that the zeros of this characteristic polynomial are the
eigenvalues of $F_{p}$. Recall that,

\begin{definition}
\label{stable} The AR model of order-$p$ in \eqref{AR} is stable (wide
sense-stationary) if all of the roots of the characteristic polynomial in
\eqref{Pi} (or equivalently, all of the eigenvalues of matrix $F_{p}$) are in
the interior of the unit circle in a complex plane, that is, solutions of
$\Pi(x)=0$ satisfy $|x|<1$.
\end{definition}

\subsection{Classical parameterization method}

Given time series $\{u_{m}\}_{m=1,\ldots,M}$, a classical method to
parameterize AR model in \eqref{AR} is with a linear regression based
technique known as the Yule-Walker method \cite{brockwell:02}. The Yule-Walker
estimators for $\mathbf{a}=(a_{1},\ldots,a_{p})^{\top}$ and $Q$ are given as follows,
\begin{align}
\hat{\mathbf{a}}  &  = \arg\min_{\mathbf{a}} \| \mathbf{y}-X(\mathbf{a}%
+\mathbf{e}_{p})\|_{2},\label{lsq}\\
\hat{Q}  &  = \frac{\|\mathbf{y}-X(\mathbf{a}+\mathbf{e}_{p})\|_{2}^{2}}{M-p},
\label{res}%
\end{align}
where $\mathbf{e}_{p}$ denotes a $p$-dimensional vector that is one on the
last component and zero otherwise. In \eqref{lsq} and \eqref{res}, $X$ is an
$(M-p)\times p$ matrix for which the $j$th row is $\mathbf{u}_{j}^{\top}-
\bar{u}\mathbf{1}^{T}$ and $\mathbf{y}=(u_{p+1}-\bar{u},\ldots,u_{M}-\bar
{u})^{\top}$, where $\bar{u}$ is the temporal average of $u_{m}$ that is used
as an estimator for the forcing term $f$ in \eqref{AR} and $\mathbf{1}$ is a
$p$-dimensional vector with unit components. The parameter $p$ can be
obtained from the Akaike criterion \textbf{(AIC)} \cite{akaike:69}, which chooses
$p$ that minimizes $\mathcal{F}(p)=\hat{Q}(M+p)/(M-p)$. This method requires
one to compute the Yule-Walker estimators in \eqref{lsq}, \eqref{res} for
various $p$, which can be time consuming. In our experience, the resulting AR
model is statistically accurate only when the given data set is large, $M\gg1$
and its statistical accuracy is sensitive to the sampling time, $\delta t$, as we will see in
Section~4 below. For different parameterization methods, we refer to the discussion in \cite{neumaierschneider:01}.


\subsection{Linear multistep methods}

Consider approximating solutions of the following linear initial value problem
(IVP),
\begin{align}
\frac{du}{dt}  = \lambda u + F,\quad u(0)  = u_{0},\label{test}
\end{align}
where $\lambda<0$ and $F$ denotes a constant forcing, with a linear multistep method of order-$p$ and discrete integration time step $\delta t$. 
The resulting discrete recursive equation can be described by the deterministic term of the AR model in \eqref{AR} for an appropriate choice
of $a_{j}$. First, let us state the consistency conditions for the AR model in \eqref{AR} that are equivalent to the consistency conditions of the explicit Adams-Bashforth scheme \cite{quarteroni:07,bh:13b}:

\begin{proposition}
The $p$-th AR model \eqref{AR} for approximating the linear test
model in \eqref{test} is consistent of order-$q$, for $1\leq q\leq
p$, if parameters $\{a_{j}\}_{j=1,\ldots,p}$ in \eqref{AR} satisfy:
\begin{align}
\ell\sum_{j=1}^{p} (j-p)^{\ell-1} a_{j} = \lambda\delta t, \quad \ell=
1,\ldots,q. \label{consistency}%
\end{align}
\end{proposition}
\begin{proof}
See Appendix A.
\end{proof}

When coefficients $\{a_{j}\}_{j=1,\ldots,p}$ of an order-$p$ model are chosen
to satisfy the consistency conditions in \eqref{consistency} of exactly
order-$p$, the resulting multistep method is exactly the explicit 
Adams-Bashforth scheme of order-$p$ that approximates the solutions of the linear ODE in \eqref{test}.
Generally, if the linear multistep method is also stable in the sense of Definition~\ref{stable}, then
the approximate solutions with initial conditions, $\{u_{j}\}_{j=m-p+1}^m$ that tend to $u_{0}$ as $\delta
t\rightarrow0$, converge to the exact solutions of the IVP in \eqref{test}. In
particular, their difference at a fixed time is of order-$p$. The converse
statement is also true. Note that this convergent result \cite{quarteroni:07} is similar to
the fundamental (Lax-equivalence) theorem in the analysis of finite difference
methods for numerical approximation of PDEs.

\subsection{Approximate Kalman filtering with AR model}

Consider filtering noisy observations of $u(t)$ at discrete time step,
$t_{k+1}-t_{k} = \Delta t = n\delta t$ for fixed $n$,
\begin{align}
du  &  = (\lambda u + F)\,dt + \sigma\,dW,\label{perfect}\\
v_{k}  &  = u(t_{k}) + \varepsilon_{k}, \quad\varepsilon_{k}\sim
\mathcal{N}(0,R), \label{observations}%
\end{align}
where $dW$ and $\varepsilon_k$ are independent white noises. Let $u^{+}_{k}$ and
$c^{+}_{k}$ be the posterior mean and variance estimates, respectively,
obtained from the Kalman filter equations with the perfect model in
\eqref{perfect}. Mathematically, these statistics are the first two moments of
the conditional distribution of $u$ at time $t_{k}$ given all observations up
to the current time, $\{t_{j}\}_{j=0,1,\ldots,k}$, obtained from solving
Bayes theorem,
\begin{align}
p(u_{k}|v_{0},\ldots,v_{k})\propto p(u_{k}|v_{0},\ldots,v_{k-1}) \times
p(v_{k}|u_{k}). \label{bayes}%
\end{align}
In short, filtering is a sequential method for updating the prior
distribution, $p(u_{k}|v_{0},\ldots,v_{k-1})$, with the observation likelihood
function, $p(v_{mn}|u_{mn})$. Note that the Kalman solutions are optimal in
the sense of minimum variance \cite{kalman:61}.

Suppose that we have no access to the underlying dynamics in
\eqref{perfect} as in many applications. Let's consider an AR model of
order-p in \eqref{AR} as the filter model. In particular, the approximate
Kalman filter problem with the AR model of order-$p$ \cite{kh:12b,bh:13b} for
the observations in \eqref{observations} is given by,
\begin{align}
\tilde{\mathbf{u}}_{m+1}  &  = F_{p}\tilde{\mathbf{u}}_{m} + \mathbf{f} +
\mathbf{e}_{m+1}, \quad\mathbf{e}_{m+1}\sim\mathcal{N}(0,Q),\label{AR2}\\
v_{k}  &  = u(t_{k}) + \varepsilon_{k} \approx G \tilde{\mathbf{u}}_{k} +
\tilde{\varepsilon}_{k}, \quad\tilde{\varepsilon}_{k}\sim\mathcal{N}%
(0,R).\nonumber
\end{align}
Here, the observation operator $G=[0,\ldots,0,1]$ maps the temporally
concatenated vector $\tilde{\mathbf{u}}_{k} =(\tilde{u}_{k-p+1}\ldots
,\tilde{u}_{k-1},\tilde{u}_{k})^{\top}$ to the observations $v_{k}$. With this
approximate observation model, we commit model error through the AR model in
\eqref{AR2}. For simplicity, we assume that the observations are collected at
every $n$ times of the integration time step, that is, $\Delta t = n\delta t$,
where $kn=m$. For this approximate AR filter, the Kalman filter solutions are
given by the following recursive equations,
\begin{align}
\tilde{u}^{+}_{k}  &  = G\tilde{\mathbf{u}}^{+}_{k} = G[\tilde{\mathbf{u}}%
^{-}_{k} + K_{k} (v_{k}- G\tilde{\mathbf{u}}^{-}_{k})],\label{postmean}\\
\tilde{C}^{+}_{k}  &  = (\mathcal{I}-K_{k}G) \tilde{C}^{-}_{k},\label{postcov}%
\\
K_{k}  &  = \tilde{C}^{-}_{k} G^{\top}(G\tilde{C}^{-}_{k}G^{\top}%
+R)^{-1},\label{kg}\\
\tilde{\mathbf{u}}^{-}_{k}  &  = F_{p}^{n} \tilde{\mathbf{u}}^{+}_{k-1} +
\sum_{j=0}^{n-1} F_{p}^{j} \mathbf{f}, \label{priormean}\\
\tilde{C}^{-}_{k}  &  = F_{p}^{n} \tilde{C}^{+}_{k-1} (F_{p}^{\top})^{n} +
\sum_{j=0}^{n-1} F_{p}^{j} Q (F_{p}^{\top})^{j}, \label{priorcov}%
\end{align}
where $\tilde{\mathbf{u}}^{+}_{k}$ and $\tilde{C}^{+}_{k}$ denote the
posterior mean and covariance statistics, respectively, and $\tilde
{\mathbf{u}}^{-}_{k}$ and $\tilde{C}^{-}_{k}$ denote the prior mean and
covariance statistics, respectively. Notice that in the posterior mean update
in \eqref{postmean}, we multiply the posterior solutions with matrix $G$ to
ensure that we only use the incoming observations to update the mean estimates
at the corresponding time, and not the estimates at previous times. We should
also point out that the Kalman gain formula in \eqref{kg} involves only a
scalar inversion at each iteration.

\subsection{Stable and consistent AR filter in assimilating nonlinear time
series}

The main result in \cite{bh:13b} can be summarized as follows: \textit{For any
$p\geq2$, the approximate mean estimate $\tilde{u}^{+}_{k}$ from the AR filter
of order-$p$ in \eqref{AR2} is as accurate as the optimal mean estimate
$u^{+}_{k}$ of the linear filtering problem in
\eqref{perfect}-\eqref{observations} for large $n$ when the system noise
variance $Q=\sigma^{2}\delta t$ is chosen based on Euler's discretization and
parameters $\{a_{j}\}_{j=1,\ldots,p}$ are chosen to satisfy both the stability
condition in the sense of Definition~1 and the consistency condition of only
order-two of the Adams-Bashforth method as defined in Definition~2.}

Based on this theoretical result, it becomes interesting to see whether
similar result can be achieved in a nonlinear setting with an AR model that
satisfies both the stability and consistency conditions of order-two of Adam-Bashforth 
method conditions. For convenience, we define:

\begin{definition}
\label{consistent} An AR model of order-$p$ in \eqref{AR} is 
consistent if its coefficients, $a_{j}$, satisfy the consistency of order-2
in \eqref{consistency} for a given $\lambda$, $\delta t$, and $p$.
\end{definition}

For general nonlinear problems, coefficient $\lambda$ in \eqref{consistency} corresponds to the decaying time scale (linear dispersion relation) of the
time series while $\sigma$ corresponds to the noise amplitude strength. These parameters can be inferred from the energy and correlation time statistics of the underlying signals that are typically measured in many applications. Loosely speaking, the two consistency constraints in \eqref{consistency} provide some ``physical constraints" on the
purely statistical, AR model. When we have no access to the underlying dynamics of the process as in certain applications, we hope that this
``physics" constrained, AR model can provide reasonably accurate surrogate prior statistics for short term prediction as well as for data assimilation application as in the linear setting \cite{bh:13b}.

However, given $\lambda$ from nonlinear time series, it is nontrivial to determine whether a stable and consistent AR model in the sense of Definitions~\ref{stable} and \ref{consistent} even exists for a fixed-$p$ and discrete integration time step $\delta t$\footnote{We will refer to $\delta t$ as the integration time in the remaining of this paper. In Section~4, we will compare our method to a classical regression-based AR model which parameters are obtained from fitting to a training data set at this time interval, so $\delta t$ can be also called the sampling time.}. For a fixed $p$ and $\lambda$, a naive numerical approach is to find parameters $a_{j}$ by solving the linear regression problem in \eqref{lsq} subjected to two linear constraints in \eqref{consistency}. Then, use the residual formula in \eqref{res} to parameterize $Q$. In \cite{bh:13b}, they applied this approach on time series from the truncated Burgers-Hopf model \cite{mt:00} and found that the resulting AR model produces more accurate filtered solutions compared to those obtained with the corresponding AR model without the two consistency constraints. Unfortunately, their result is not robust. In another numerical example with time series from the Lorenz-96 model \cite{lorenz:96}, they reported that they can't even find parameters $a_{j}$ that yield a stable AR model when the consistency constraints are imposed. In the next section, we will devise a new algorithm based on algebraic geometry tools to avoid such issues. We will require our parameterization method to provide a range of integration time step intervals that guarantees the existence of stable and consistent AR models. Furthermore, we will require the new scheme to only take $\lambda, p, \sigma$ as inputs, as opposed to the linear regression based method in \eqref{lsq} which requires large data sets. 

\section{An algebraic-based parameterization method}
\label{sec:algebraic}
In this section, we present an algebraic method for finding stable and consistent AR models of order-3. We will divide the discussion into three subsections. In the first subsection, we will construct the boundary of the set of parameters that are stable and consistent in the sense of Definitions~\ref{stable} and \ref{consistent}, respectively. In the second subsection, we will describe a method to determine a single set of parameters for AR models of order-3 that lies in this set. In the third subsection, we will summarize the ideas into a self-contained algorithm and illustrate it on an example corresponding to the most energetic, Fourier mode-8 of the Lorenz-96 model in a weakly turbulent regime with forcing constant $F=6$ \cite{mag:05}. 

Most of the underlying ideas discussed in this section can be, in principle, applied to arbitrary order-$p$, where $p>3$. However, we restrict our discussion to the order-3 problem, for the following three reasons:
(i) It seems that the order-3 is already quite useful as we shall see in Section~\ref{sec4} below; 
(ii) The algorithm for the order-3 can be described succinctly. 
      The higher order cases are expected to be very complicated,
      involving various complex algebraic operations; 
(iii) The computing time for the order-3 is within practical range. 
      The higher order cases would require dramatically longer computing times,
      making them impractical at the current state of the art.

Let $\lambda$ be a given complex number corresponding to the linear operator of the model in \eqref{perfect}. We would like to find $a_{1}%
,a_{2},a_{3}$ and~$\delta t$ such that the AR model of order-3 is stable and consistent in the sense of Definitions~\ref{stable} and \ref{consistent}, respectively. First, let us describe this problem mathematically. 

Let us ensure the consistency of order-2 in the sense of Definition~\ref{consistent}, that is,
\begin{align*}
\lambda\delta t &=1\cdot\sum_{j=1}^{3}(j-3)^{1-1}a_{j}  = a_{1}+a_{2}+a_{3}, \\
\lambda\delta t &= 2\cdot\sum_{j=1}^{3}(j-3)^{2-1}a_{j} =-4a_{1}-2a_{2}.
\end{align*}
By solving the two equations for $a_{1},a_{2}$ and $a_{3},$ we get the
following infinitely many solutions,
\begin{align}
a_{1}  &  =\left(  s-\frac{3}{2}\right)  \lambda\delta t,\nonumber\\
a_{2}  &  =-\left(  2s-\frac{5}{2}\right)  \lambda\delta t, \label{param}\\
a_{3}  &  =s\lambda\delta t,\nonumber
\end{align}
where $s$ is an arbitrary complex number. Thus the AR model of order-3 with
the above values of $a_{1},a_{2},a_{3}$ is consistent for arbitrary values of
$s$ and $\delta t.$

Let %
\begin{align*}
\Pi\left(  s,\delta t,x\right)   &  =a_{1}+a_{2}x+a_{3}x^{2}+x^{2}-x^{3}\\
&  =\left(  s-\frac{3}{2}\right)  \lambda\delta t-\left(  2s-\frac{5}%
{2}\right)  \lambda\delta tx+s\lambda\delta tx^{2}+x^{2}-x^{3}.%
\end{align*}
From Definition~\ref{stable}, we need to ensure that the roots of $\Pi$ are within the unit circle in the complex plane,
\begin{align}
\forall x\in\mathbb{C},\;\;\;\Pi(s,\delta t,x)=0\;\Longrightarrow\;|x|<1.\label{cond}
\end{align}

\subsection{Constructing the boundary of the set of stable and consistent parameters} 
In the following, we will go through a series of steps to find the boundary of the stable and consistent subset of parameters $s, \delta t$.
In particular, let $s=\alpha+$ $\beta i.$ Then the condition in \eqref{cond} defines a subset of
the three dimensional real space for $\alpha,\beta,\delta t.$ We will refer to this subset as the {\em stable and consistent set}. 

By the continuity of the roots of $\Pi$, the following condition holds at the boundary%
\[
\exists x\in\mathbb{C},\mbox{ such that }\Pi(\alpha+\beta i,\delta t,x)=0\ \text{and }|x|=1.
\]
This condition is equivalent~\footnote{In general this is not a completely equivalent condition because the rational parameterization misses the point $x=1$. However,  $x=1$ is a root of $\Pi$ only when $\delta t=0$, so this case is not relevant to our problem.} to%
\[
\exists q\in\mathbb{R},\mbox{ such that }u:=\Pi(\alpha+\beta i,\delta t,\frac{1-q^{2}}{1+q^{2}}+\frac{2q}{1+q^{2}}i)=0.
\]
Here, we used the standard rational parameterization of a unit circle
(based on half-tangent). Since the denominator of $u$ is a certain power of $1+q^{2},$ it is
never zero. Thus the above condition is equivalent to
\BEA
\exists q\in\mathbb{R,}\mbox{ such that }g:=\text{numerator of}\;u=0.\label{cond1}
\EEA
Note that the equation $g=0$ is over the complex variable, and hence it is
actually two equations over the real variables. Hence we can rewrite the condition in \eqref{cond1} as%
\[
\exists q\in\mathbb{R},\mbox{ such that }\operatorname{Re}g=0\ \text{and }\operatorname{Im}%
g=0,
\]
where 
$\operatorname{Re}g$ and $\operatorname{Im}g$ are polynomials with real coefficients.

We can eliminate the existentially quantified variable $q$, by finding a generator $r$ 
for the elimination ideal of 
$
\left\langle \text{Re }g,\text{Im } g\right\rangle
$
over $\alpha,\beta,\delta t$.
Informally speaking, the  generator $r$ is a polynomial in the variables 
$\alpha,$ $\beta$ and $\delta t$ so that the solution set of $r=0$ contains the projections of the solution set 
of  the system of equations 
$\operatorname{Re}g=0\ \text{and }\operatorname{Im}%
g=0$ onto the $(\alpha,\beta,\delta t)$ plane.
In a sense, we have eliminated the variable $q$.
For a concise and precise definition of elimination ideal, see  Appendix~B.
For the details, see the highly readable undergraduate textbook on computational algebraic geometry~\cite{clo:10}.
 
Note that $r=0$ defines a surface in the three dimensional real space
for $\alpha,\beta,\delta t$ (see Figure~\ref{fig_r3dr2d}). 
This surface contains the boundary of the stable and consistent set.
For a given value of $\delta t,$ by sampling and checking, one can
determine that the stable and consistent set is the one enclosed by the contour curve for
$\delta t$ (see Figure~\ref{fig_stable}).

From the three-dimensional surface, there are obviously non-unique choices of parameters $\alpha, \beta, \delta t$ that will lie in the stable and consistent set. In the following, we will describe a method for choosing one of these parameters.

\subsection{Choosing a set of stable and consistent parameters}

The contour plot in Figure~\ref{fig_r3dr2d} seems to indicate that the stable and consistent set of a smaller $\delta t$
contains the stable and consistent set of a larger $\delta t.$ Careful computation shows that
it is almost always true. It is violated only when $\delta t$ is very close to the
extreme (where there is no stable and consistent set). Hence it motivates us to choose
$(\alpha,\beta)  $ as the one which is contained in all the
stable and consistent sets for almost all $\delta t$ values. Such a point 
$(\alpha,\beta)$ is indicated by a black cross on the contour plot in Figure~\ref{fig_r3dr2d}. Let us
denote it by ($\hat\alpha ,\hat\beta$) and the corresponding
$\delta t$ as $\hat{\delta t}.$

Note that the point 
$(\hat{\alpha},\hat{\beta},\hat{\delta t})$ lies on the self-intersection curve of the surface $r=0$ (see the plot on the left in Figure ~\ref{fig_r3dr2d}). Thus it is a singular point of the
surface, in other words, it satisfies the following system of equations,%
\BEA
r=0, \quad\frac{\partial r}{\partial\alpha}=0, \quad \frac{\partial r}{\partial\beta}%
=0, \quad\frac{\partial r}{\partial\delta t}=0.
\label{systemr}
\EEA

Here we face a difficulty, the system of equations in \eqref{systemr}
has infinitely many solutions, that is, there are infinitely many singular
points, consisting of the point of self-intersection of each contour (see Figure~\ref{fig_r3dr2d}). These are
not of interest. Hence we only need to find the isolated solutions of the
above system of equations. This can be done, for instance, by carrying out the following
algebraic operations,%
\begin{align*}
J  &  =\text{generators of the elimination ideal of}
\left\langle r,
\frac{\partial r}{\partial\alpha},
\frac{\partial r}{\partial\beta},
\frac{\partial r}{\partial \delta t}
\right\rangle
\text{over }\alpha,\beta,\\
D  &  =\text{prime  decomposition of }\left\langle J \right\rangle\\
W  &  =\bigcup\limits_{\substack{d\in D\\\text{dim}\left(  d\right)
=0}}\ \left\{  \ \text{real solutions of }d\ \right\},
\end{align*}
where $W$ contains all the isolated singular points on the $(\alpha$,$\beta$) plane.

Informally speaking, the elimination ideal $J$ is a set of polynomials in the variables $\alpha$ and $\beta$ so that the solution set of $J=0$ contains the projections of the solution set of  the system of equations 
$r=0, 
\frac{\partial r}{\partial\alpha}=0, 
\frac{\partial r}{\partial\beta}=0,
\frac{\partial r}{\partial\delta t}=0$ 
onto the $(\alpha,\beta)$ plane.
In a sense, we have eliminated the variable $\delta t$.  The prime decomposition $D$ is a set of polynomials, that is  $D=\{d_1,d_2,\ldots\},  $ so that the union of the solution sets of $d_1=0,d_2=0,\ldots$ is the same as the solution set of $J=0$ and that each $d_i$ is ``prime'' in a sense similar to prime number. Hence, it can be viewed as a generalization of prime factorization of integers to the system of polynomial equations. The notation $\text{dim}(d)$ stands for the dimension of the solution set of $d=0$.
Hence $\text{dim}(d)=0$ states that the solution set of $d=0$ is a finite set (consisting of finitely many isolated points). 
For concise and precise definitions of elimination ideal and prime decomposition, see Appendix~B.
For the details, see the highly readable undergraduate textbook on computational algebraic geometry~\cite{clo:10}.

Now we need to choose a suitable point $(\hat{\alpha},\hat{\beta
})  $ from $W.$ For a given $(\alpha,\beta)  \in W,$ the
polynomial equation $r(  \alpha,\beta,\delta t)  =0$ has finitely
many (positive) real solutions for $\delta t.$ Thus, the AR model of order-3
is stable for $\delta t\in (0,\bar{\delta}t)  $ where,
\[
\bar{\delta}t=\min\limits_{\substack{\delta t>0\\r\left(  \alpha,\beta,\delta
t\right)  =0}}\delta t.
\]
Obviously, $\bar{\delta}t$ depends on $(\alpha,\beta).$ For practical consideration, we prefer larger~$\bar{\delta}t.$ Thus we choose $(\alpha,\beta) \in W$ so that $\bar{\delta}t$ is maximum and obtain
 \begin{align*}
 (\hat{\alpha}, \hat{\beta} ) &  =\arg\max\limits_{\substack{\alpha,\beta\\\left(  \alpha
,\beta\right)  \in W}}\min\limits_{\substack{\delta t>0\\r\left(  \alpha
,\beta,\delta t\right)  =0}}\delta t,
\end{align*}
and the maximum value is at $\hat{\delta}t$.

Let $\hat{s}=\hat{\alpha}+\hat{\beta}i$ and substituting this to \eqref{param}, we obtain:%
\begin{align*}
a_{1}  &  =\left(  \hat{s}-\frac{3}{2}\right)  \lambda\delta t,\\
a_{2}  &  =-\left(  2\hat{s}-\frac{5}{2}\right)  \lambda\delta t,\\
a_{3}  &  =\hat{s}\lambda\delta t.
\end{align*}
Then we have
\[
\forall\,\delta t\in (  0,\hat{\delta t}),\quad \forall
x\in\mathbb{C},\quad\Pi(s,\delta t,x)=0\;\Longrightarrow\;|x|<1,
\]
that is, $\forall\,\delta t\in(0,\hat{\delta t})  $ the AR model of order-3 is stable and consistent.

\subsection{Algorithm}\label{sec33}
Table~\ref{algorithm} provides 
 a self-contained algorithm summarizing the ideas discussed in the previous
subsection. We illustrate the algorithm on an example.

\begin{table}[pt]
\caption{Algebraic method for parameterizing stable and consistent AR models of order-3.}
\label{algorithm}
\medskip
\hrule

\begin{description}\itemsep=0pt
\item[In:\hfill] $\lambda\in\mathbb{C}$, such that $\text{Re\ }\lambda<0$
\item[Out:] $a_{1},a_{2},a_{3}\in\mathbb{C}[\delta t]$ and $\hat{\delta t}
\in\mathbb{R}^{+}$, such that for all $\delta t \in(0,\hat{\delta t})$, the AR
model of order-3 is stable and consistent.
\end{description}

\begin{enumerate}\itemsep=0pt
\item \label{step_h} $\Pi=\left(  s-\frac{3}{2}\right)  \lambda\delta
t-\left(  2s-\frac{5}{2}\right)  \lambda\delta tx+s\lambda\delta tx^{2}%
+x^{2}-x^{3}$

\item \label{step_u} $u=$ substitution$(\Pi;s=\alpha+\beta i,x=\frac{1-q^{2}%
}{1+q^{2}}+\frac{2q}{1+q^{2}}i)$

\item \label{step_g} $g=\text{numerator of}\;u$

\item \label{step_r} $r=$ a generator of the elimination ideal of 
                     $\left\langle\text{Re } g,\text{Im } g\right\rangle$ over $\alpha,\beta,\delta t$

\item \label{step_J} $J=$ generators of the elimination ideal of $\left\langle r,\frac{\partial r}{\partial\alpha},\frac{\partial
r}{\partial\beta},\frac{\partial r}{\partial\delta
t}\right\rangle $ over $\alpha,\beta$

\item \label{step_D} $D=\text{prime  decomposition of} \left\langle J \right\rangle$

\item \label{step_W} $W=\bigcup\limits_{\substack{d\in D\\\text{dim}\left(
d\right)  =0}}\ \left\{  \ \text{real solutions of }d\ \right\}  $

\item \label{step_dt} $\hat{\delta t}=\max\limits_{\substack{\alpha
,\beta\\\left(  \alpha,\beta\right)  \in W}}\min\limits_{\substack{\delta
t>0\\ r\left(  \alpha,\beta,\delta t\right)  =0}}\delta t$

\item \label{step_alpha_beta} $( \hat{\alpha},\hat{\beta})
=\arg\max\limits_{\substack{\alpha,\beta\\\left(  \alpha,\beta\right)  \in
W}}\min\limits_{\substack{\delta t>0\\r\left(  \alpha,\beta,\delta t\right)
=0}}\delta t$

\item \label{step_shat} $\hat{s}=\hat{\alpha}+\hat{\beta}i$

\item \label{step_a} $a_{1}=\left(  \hat{s}-\frac{3}{2}\right)  \lambda\delta
t$

\item[] $a_{2}=-\left(  2\hat{s}-\frac{5}{2}\right)  \lambda\delta t$

\item[] $a_{3}=\hat{s}\lambda\delta t$
\end{enumerate}
\hrule
\end{table}

\begin{example} \

\begin{description}
\item[In:] $\lambda= -8.312 - 8.569i$. This choice of $\lambda$ corresponds to the linear dispersion relation of the most energetic Fourier mode-8 of the Lorenz-96 model in weakly chaotic regime \cite{mag:05} (this can be obtained from a simple linear regression Mean Stochastic Model fit as proposed in \cite{mgy:10,mh:12}, cf. \eqref{msm}). 
\end{description}

\begin{description}
\item [\bf Out:] 
\begin{itemize}
\item $a_{1} =\left(  -0.251+3.147i\right)  \delta t$
\item $a_{2} =\left(  4.657-2.010i\right)  \delta t$
\item $a_{3} =\left(  -12.718-9.706i\right)  \delta t$
\item $\hat{\delta t}=0.145$
\end{itemize}
\end{description}

\noindent Figure~\ref{fig_eig} shows the complex roots of $\Pi(x)=a_{1}+a_{2}x+a_{3}x^{2}+x^{2}-x^{3}$ for several values
of $\delta t \in[0,\hat{\delta t}]$. When $\delta t =0$ or $\hat{\delta t}$,
$\Pi$ has one complex root on the unit circle. When $\delta t \in
(0,\hat{\delta t})$, all the complex roots of $\Pi$ lie inside of the unit circle.%

\end{example}

\section{Numerical results}\label{sec4}

In this section, we numerically verify the statistical accuracy of the AR
models obtained from the proposed parameterization scheme for state estimation
and short-time prediction of nonlinear time series. In particular, we will
compare the numerical results based on the following AR models:

\begin{itemize}
\item A standard AR model, parameterized with Yule-Walker estimators in
\eqref{lsq}-\eqref{res}. Here, parameter $p$ is empirically
chosen based on AIC criterion for a fixed $\delta t$ as discussed in Section~2.2. This is a standard
method that was proposed in \cite{kh:12b}. We will refer to this model as the
\textbf{AR-p model}.

\item The second AR model is basically a standard AR model with two
consistency constraints in \eqref{consistency}. The model parameters are
obtained from solving the minimization problem in \eqref{lsq} subjected to two
linear constraints in \eqref{consistency}. Then the Yule-Walker estimator in
\eqref{res} is used to determine $Q$. Here, we simply fix $p$ and $\delta t$
based on the AIC criterion obtained from the unconstrained AR-p model above.
We refer to this model as the consistent AR-p model (in short, \textbf{CAR-p model}).

\item To diagnose the impact of the consistency constraints in
\eqref{consistency}, we also consider a standard AR model with fixed $p=3$
and $\delta t$, parameterized with Yule-Walker estimators. We refer to this
model as the \textbf{AR-3 model}.

\item The proposed AR model with fixed $p=3$ and $\delta t$, parameterized
based on the algorithm discussed in Section~3. Here, the system noise variance
is determined by Euler discretization, $Q=\sigma^{2}\delta t$. We refer to
this model as the stable and consistent AR-3 model (in short, \textbf{SCAR-3 model}). One feature of
this model is that the parameters can be obtained without knowing the
signal time series directly. This model simply takes $\lambda, \sigma$, which
can be inferred from measurements of the energy $\mathcal{E}$ and the decaying
time scale, $\mathcal{T}$, of the signals through linear regression Mean
Stochastic Model proposed in \cite{mgy:10,mh:12},
\begin{align}
\lambda= \mathcal{T}^{-1}, \quad\sigma^{2} = 2\, \mbox{Re}[\lambda] \mathcal{E}.\label{msm}
\end{align}
Practically, one determines the parameters by choosing an integration time, $\delta t\in(0,\hat{\delta}t)$,
based on $\hat{\delta} t$ obtained from the algorithm in Table~\ref{algorithm}.
\end{itemize}

We use the temporal average Root-Mean-Square Error (RMSE)
difference between the estimates, $\hat{u}_{k}$, and the truth, $u_{k}$, defined
as,
\begin{align}
E(\hat{u}) = \Big( \frac{1}{T} \sum_{k=1}^{T} |\hat{u}_{k} - u_{k}|^{2}
\Big)^{1/2},\nonumber
\end{align}
to quantify, both, the short-time forecasting and filtering skills. In
particular, we use the prior estimate error, $E(\hat{u}^{-})$, to determine
the accuracy of the short-time prediction skill. Similarly, we use the
posterior estimate error, $E(\hat{u}^{+})$, to quantify the filtering skill.

\subsection{Application on the Lorenz-96 model: A toy example}
As a testbed, we consider time series from the 40-dimensional Lorenz-96 model \cite{lorenz:96},
\begin{align}
\frac{dx_j}{dt} =(x_{j+1}-x_{j-2})x_{j-1} -x_j+F, \quad j=1,\ldots, 40,\nonumber
\end{align}
in a weakly turbulent regime with forcing constant $F=6$ \cite{kh:12b,mag:05}. 
In particular, we report numerical results on Fourier wave numbers 8 and 1, which have distinct
characteristics. For these two modes, the resulting AR-p models based on AIC criterion have
$p=15$-lags, which are empirically tuned with sampling time $\delta t=4/16$ (see \cite{kh:12b} for details). 
For mode-8, we expect the AR-p model to excel since the underlying signal has longer memory with very slow
oscillatory autocorrelation function (see Figure~\ref{fig3_0}). For mode-1, we
expect the AR-3 model to be very accurate since the autocorrelation
function of this signal decays very quickly compared to mode-8 (again, see
Figure~\ref{fig3_0}).

\subsubsection{Results on Mode-8}

For Fourier mode-8, we obtain $\lambda=-8.312 - 8.569\text{i}$ and our algorithm suggests that the stable and consistent SCAR-3 models can be constructed with integration times $0\leq\delta t< \hat\delta t = 0.145$. 

In Figure~\ref{fig3_1}, we show the distribution of the eigenvalues of the four AR models at nine integration time steps, $\delta t$, all of which produce stable and consistent AR-3 models as shown: 1) AR-p model (magenta triangles); 2) CAR-p model (red diamonds); 3) AR-3 model (blue circles); 4) SCAR-3 model (black plus sign). Notice that at smaller integration times, $\delta t = 1/64, 2/ 64$, the AR-p models are unstable (some eigenvalues are not in the unit circle. On the other hand, for larger
integration times, $\delta t > 5/64$, the CAR-p models are unstable. This confirms the difficulties of finding an AR model that is stable and consistent, as reported in \cite{bh:13b}. Here, the AR-3 models are always stable and the eigenvalues are distributed near $(1,0)$ of the unit circle when the integration times are small. As the integration time increases, the eigenvalues tend to spread out near the boundary of the unit circle. On the other hand, the eigenvalues of the SCAR-3 model do not cluster about one point for any shown integration time. As the integration time increases, one of the eigenvalues tends to be around $(0.5,0)$, whereas the other two eigenvalues are near the boundary of the unit circle in the third quadrant.

In Figure~\ref{fig3_3}, we show the average RMSE of the posterior estimates as functions of integration time, $\delta t$ (on an increment of 1/64), for various observation times $\Delta t=n\delta t$, where $n=1, 10, 50$, and observation noise variances, $R=10\%\mathcal{E}, 25\%\mathcal{E} , 50\%\mathcal{E},100\%\mathcal{E}$. Notice that the estimates from the AR-p model (magenta triangle) tends to blow up in various regimes, especially for smaller integration times. The AR-3 model (blue circle) performs slightly better than the AR-p model, although it also blows up occasionally. This result demonstrates the sensitivity of the statistical estimates of the regression-based AR model (as discussed in Section~2.2). The CAR-p analysis estimates (red diamond) are comparable to the SCAR-3 model in almost every regime except for large observation time, $n=50$, and long integration time steps, $\delta t>6/64$; this divergence is not surprising since the model is very unstable in this regime. In terms of prior estimates, this divergence looks more pronounced even for $n=10$ (see Figure~\ref{fig3_4}). The posterior and prior estimates obtained from the proposed, SCAR-3 model, are very accurate; their average RMSE are consistently small, below the observation error. This result suggests that the two consistency constraints in \eqref{consistency} and the stability condition in Definition~\ref{stable} provide robustly accurate short-term prediction for the AR-3 model.

\subsubsection{Results on Mode-1}

For mode-1, we obtain $\lambda=-1.246-1.214\text{i}$ and our algorithm suggests that SCAR-3 models can be constructed with integration times chosen on interval, $0\leq\delta t<\hat\delta t = 1.006$. In Figure~\ref{fig3_6}, we show the average RMSE of the posterior and prior estimates as functions of integration time, $\delta t$ (on an increment of 4/64), for various observation noise variances, $R$, and observation times with $n=10$. In this regime, we learn that the filtered posterior estimates from all the four AR models have comparable accuracy (see the first column of Figure~\ref{fig3_6}). For a very short integration time, both the unconstrained linear regression based filtered estimates diverge to infinity in finite time.

In terms of predictive skill, the prior estimate errors of the AR-p model can blow up for smaller integration times. The prior estimate errors of the CAR-p model are large (increase to order 10) for larger integration times; this is because these consistent AR-p models are not stable for large $\delta t$. The prior estimate errors from the AR-3 model are sometimes larger than the observation error (or even blow up) when the integration time is small; for larger integration times, $\delta t>4/64$, however, their prior estimates are the most accurate. This is not surprising since the autocorrelation function of the underlying signal decays quickly and thus the underlying signal can be accurately modeled with an AR model with smaller lag $p$. The proposed SCAR-3 model produces prior estimate errors that are smaller than the observation error, except when the integration times are close to $\delta t=1$ for $R=10\%\mathcal{E}$; however these errors (on the order of $10^{-1}$) are much smaller than the largest errors produced by the other three models. These numerical results suggest that the proposed, stable and consistent, SCAR-3, model produces robust filtering and short-time predictive skill for signals with autocorrelation function that decays quickly.

\subsection{Application on predicting RMM indices: A real-world example}

In this section, we show numerical results in predicting a data set, known as Realtime Multivariate MJO (RMM) index \cite{wh:04} which is used to characterize the variability of a dominant wave pattern observed in the tropical atmosphere, the Madden Julian Oscillation (see e.g. \cite{zhang:05} for a short review). The given data set is a two-dimensional time series, obtained from applying empirical orthogonal functions analysis on combinations of equatorially averaged zonal wind at two different heights and satellite-observed outgoing long wave radiation \cite{wh:04}. The first component of the data set is denoted as RMM1 and the second component as RMM2. 

We will model these indices as a complex variable where the real part denotes RMM1 and the imaginary part denotes RMM2 since the trajectory of these pairs rotates counter-clockwise around the origin especially in strong MJO phase. To produce a reasonable ensemble forecast, we need an ensemble of initial conditions. In this application, we apply the ensemble Kalman filter algorithm developed in \cite{hunt:07}; this method basically applies Kalman filter formula, updating the empirical mean and covariance statistics produced by the ensemble forecast. We will set the ensemble size to an arbitrary choice, $50$ members. 

An additional difficulty here is that the data set looks quite noisy but we don't know the observation noise variance, $R$, nor the corresponding distribution. In fact, we don't even know whether the noises are additive or multiplicative types. In our implementation of the EnKF, we apply a simple adaptive noise estimation scheme \cite{bs:13} to extract $R$ from the innovation statistics, $\epsilon_k \equiv v_k -G \tilde{\bu}^-_k$. In our simulations, we found that $R\approx 0.02\mathcal{I}$, which means that either the noises are not additive type and Gaussian (which are implicitly assumed by the noise estimation algorithm) or the AR modeling may not be the best choice for this problem. 

From daily data set of period between Jan 1, 1980-Aug 31, 2011, we obtain $\lambda = -0.4458+3.7161i$ from setting $\delta t =12/365$ as 1 day (such that a unit denotes a month). The resulting parameters are reported in Table~\ref{tab2}. We check the AIC criterion and it shows that the optimum lag is $p\approx 2-3$ so we include the linear-regression based AR-3 model for comparison purpose. Since the time scale of this data set is short, as in the mode-1 of the Lorenz-96 example (Section~4.1.2), we expect the AR-3 to give the best estimates among the class of AR models (again, based on AIC criterion).
\begin{table}[htdp]
\caption{Parameters obtained from fitting RMM index data set for period of Jan 1, 1980-Aug 31, 2011.}
\begin{center}
\begin{tabular}{|c|c|c|}
\hline
model & AR-3 & SCAR-3 \\
\hline
$a_1$ & 0.0564-0.0679i & -0.0381 + 0.0083i \\
$a_2$ & -0.5877 + 0.1307i & 0.0836 - 0.0777i \\
$a_3$ & 0.4938 - 0.0005i & -0.0601 + 0.1916i \\
$Q$ & 0.0584 & 0.0292 \\
\hline
\end{tabular}
\end{center}
\label{tab2}
\end{table}

In Figure~\ref{pc_rmm}, we show the forecasting skill as a function of lead time (in days) with a bivariate pattern correlation 
between the mean estimate, $\hat \bu$, and observation, $\bu$, 
\begin{align}
PC(t) = \frac{\hat{\bu}(t)\cdot \bu(t)}{\| \hat{\bu}\|_2\| \bu\|_2},
\end{align}  
which was used in \cite{zhangetal:13} to evaluate the forecasting skill of RMM index. Here, the component of vectors $\hat{\bu}, \bu$, is daily data from the period of Sept 1, 2011-Aug 31, 2012, so we are evaluating the forecasting skill beyond the period of the training data set. As a reference, the pattern correlations at lead time of 15 days from many participating working groups using operational dynamical models vary between 0.4-0.8 \cite{zhangetal:13}. In this measure, the forecasting skill of both models are comparable. However, when we look at the actual mean forecast estimate (see Figure~\ref{fcst} for the SCAR-3 10-day lead forecast for two different periods), the forecast from SCAR-3 model looks more accurate, following the peaks of the noisy signals, compared to that of the AR-3 model. 

While it is unclear that the proposed model is appropriate for fitting the RMM index, this numerical result indicates a slight advantage of imposing the consistency constraints via the proposed algebraic method over the standard regression-based technique.

\section{Summary}

In this paper, we discussed a novel parameterization method for low-order linear autoregressive models for filtering nonlinear chaotic signals with memory depth. The new algorithm was constructed based on standard algebraic geometry tools. The resulting algorithm also provides an upper bound for discretization time step that guarantees the existence of stable and consistent AR models which can be an issue in nonlinear setting when regression-based method is used \cite{bh:13b}. An attractive feature of this parameterization method is that it only requires long-time average statistics as inputs whereas the classical linear regression-based method usually requires a long time series of training dataset. 

Our numerical results suggested that the short-time predictive skill of the proposed AR models is significantly more accurate than the regression-based AR models. In terms of filtering skill, the proposed models are also comparable to (or slightly more accurate than) the regression-based AR models. In our numerical test with two chaotic time series of different characteristic of time scales, we found that the new parameterization scheme is robust and stable, across wide ranges of discretization (or sampling) times, observation times, and observation noise variances. We also verify the hypothesis on a real-world problem, predicting the MJO index. 

These numerical results suggested that the two consistency constraints in the sense of Definition~\ref{consistent}, together with the stability condition in the sense of Definition~\ref{stable} can improve the short-time prediction and filtering skill with linear autoregressive prior models. In some sense, the consistency and stability conditions advocated here enforce some ``physical constraints" to the, otherwise, purely statistical AR models. However, unlike the physics constrained model proposed in \cite{mh:13,hmm:14}, the long term (equilibrium) statistical prediction of this model is not accurate at all. This is not a surprise at all since the parametric form of AR models may not be sufficient for accurate equilibrium statistical prediction as suggested by the linear theory for filtering with model error \cite{bh:14}; for the two modes example in Section~\ref{sec4}, one needs an AR model of higher order to perfectly fit the two-time equilibrium statistical quantities autocorrelation function \cite{kh:12b,khm:13}. 

Although the idea discussed in this paper, in principle, can be generalized to arbitrary order, $p>3$, we suspect that this will require more complicated algebraic operations and the resulting algorithm will be computationally impractical at the current state of the art. Furthermore, whether the same consistency constraints are useful for multivariate AR models is a wide open question. As a consequence, the method has a practical limitation if for some reason the integration time step $\delta t$ has to be fixed and it is not in the stable and consistent set of the AR-3 model, $\delta t \notin (0,\hat{\delta t})$. We should note that it is a wide open problem to design a scheme for choosing $p>3$  for fixed $\delta t$ such that the AR model is stable and consistent.


\section*{Acknowledgment}
The first author thanks Nan Chen (NYU) for helpful discussion on RMM index prediction problems and thanks Chidong Zhang (RSMAS, Miami) for sharing the data set of RMM index. The research of JH was partially supported by the Office of Naval Research Grants N00014-11-1-0310, N00014-13-1-0797, MURI N00014-12-1-0912 and the National Science Foundation grant DMS-1317919. JLR was partially supported as an undergraduate research assistant through Harlim's ONR and NSF grants. The research of HH was partially supported by NSF grant CCF-1319632.

\section*{Appendix A: Proof of Proposition~1}
To prove proposition~1, we need to only show that the consistency conditions in~\eqref{consistency} are indeed equivalent to Theorem 11.3 in \cite{quarteroni:07} {\bf for approximating linear ODE in \eqref{test}}. 

First, recall that from Theorem 11.3 in \cite{quarteroni:07}, given a general ODE, $\dot{u}=f(t,u)$, the explicit multistep method of $p$-steps,
\begin{align}
u_{m+1} = \sum_{j=0}^{p-1} \hat{a}_j u_{m-j} +\delta t \sum_{j=0}^{p-1} b_j f_{m-j}, \quad m\geq p\label{arp1}
\end{align}
where $f_m = f(t_m,u_m)$, $u_m = u(t_m)$ and $\delta t=t_{m+1}-t_m$,
is consistent and of order-$q$ when the following algebraic conditions are satisfied,
\begin{align}
\sum_{j=0}^{p-1} \hat{a}_j &= 1, \nonumber\\
\sum_{j=0}^{p-1} (-j)^i \hat{a}_j + \ell \sum_{j=0}^{p-1} (-j)^{\ell-1}b_j &= 1, \quad \ell = 1,\ldots, q. \label{cc1}
\end{align}

For Adams-Bashforth method, $\hat{a}_0 = 1$ and $\hat{a}_j=  0$ when $j\neq 0$, and \eqref{cc1} reduces to,
\begin{align}
\ell \sum_{j=0}^{p-1} (-j)^{\ell-1}b_j &= 1, \quad \ell = 1,\ldots, q. \label{sumbj}
\end{align}
In our hypothesis, we denote the AR model in \eqref{AR} as follows,
\begin{align}
u_{m+1} = u_m + \sum_{j=0}^{p-1} a_{p-j} u_{m-j}.\label{arp2}
\end{align}
Matching the notations in \eqref{arp1} and \eqref{arp2}, we have $a_{p-j}u_{m-j} = b_j f_{m-j}\delta t$. For linear ODE in \eqref{test} with $F=0$ (since we typically fit the fluctuation in AR modeling), $f_{m-j} = \lambda u_{m-j}$, and therefore, we have,
$a_{p-j} = b_j \lambda \delta t$. Substituting this to \eqref{sumbj}, we obtain
\begin{align}
\ell \sum_{j=0}^{p-1} (-j)^{\ell-1}a_{p-j} &= \lambda\delta t, \quad \ell = 1,\ldots, q,\nonumber
\end{align}
which is eqn~\eqref{consistency} when index $j$ is replaced by $p-j$.

\section*{Appendix B: Review of algebraic geometry}
\label{app:algebraic} 
In this appendix, we list precise definitions of several algebraic notions used in Section~\ref{sec:algebraic}.  
For their geometric meaning and algorithms, see the highly readable undergraduate textbook on computational algebraic geometry~\cite{clo:10}.

Let $k$ be a field, such as the set of all rational numbers. Let $R_{n}=k\left[
x_{1},\ldots,x_{n}\right]  ,$ that is, the set of all polynomials with
coefficients from $k$ and variables~$x_{1},\ldots,x_{n}.$  

\begin{definition}
[Ideal]Let $I\subset R_{n}.$ We say that $I$ is an \emph{ideal} of $R_{n}$ if
the followings hold

\begin{enumerate}
\item $0\in I.$

\item If $f,g\in I$ then $f+g\in I.$

\item If $f\in I$ and $h\in R_{n},$ then $hf\in I.$
\end{enumerate}
\end{definition}

\begin{proposition}
[Generator]Let $f_{1},\ldots,f_{s}\in R_{n}.$ Let 
\[ I=\left\{  \sum_{i=1}%
^{s}h_{i}f_{i}:h_{1},\ldots,h_{s}\in R_n\right\}.
\]  
Then $I$ is an ideal of
$R_n$. We call $I$ the ideal \emph{generated by} $f_{1},\ldots,f_{s}$ and  denote it as
$\left\langle f_{1},\ldots,f_{s}\right\rangle $. We call $f_{1},\ldots,f_{s}$
\emph{generators} of the ideal $I.$
\end{proposition}

\begin{theorem}
[Hilbert Basis Theorem]
Let $I$ be an ideal  of $R_n$. Then \[ I=\left\langle f_{1},\ldots,f_{s}\right\rangle\] 
for some finitely many $f_1,\ldots,f_s \in I$. 
\end{theorem}

\begin{proposition}
[Elimination Ideal]Let $I$ be an ideal of $R_{n}.$ Let $J=I\cap R_{i}$ for
$1\leq i<n.$ Then $J$ is an ideal of $R_{i}$. We call $J$ the
\emph{elimination ideal} of~$I$ over $x_{1},\ldots,x_{i}.$
\end{proposition}

\begin{definition}
[Prime Ideal]An ideal $I$ of $R_{n}$ is called \emph{prime }if
\[
\forall f,g\in R_{n}\ \ \ \ fg\in I\ \ \Longrightarrow\ \ f\in I~\text{or\ }%
g\in I.
\]
\end{definition}

\begin{definition}
[Prime Decomposition]Let $I$ be an ideal of $R_{n}$. Let $J_{1},\ldots,J_{s}$
be ideals of $R_{n}.$ We say that $J_{1},\ldots,J_{s}$ is a\emph{ prime
decomposition} of $I$ if $I=J_{1}\cap\cdots\cap J_{s}$ and $J_{1},\ldots
,J_{s}$ are prime.
\end{definition}



\comment{   \clearpage%
\begin{table}%
[htdp]
\begin{center}
\begin{tabular}%
{|c|c|c|c|}\hline variable & mean & variance & correlation time \\
\hline$u_{1}$ & -0.0004 - 0.0073i & 0.0901 & 13.6284 (6.3701) \\ $u_{2}$ &
-0.0051 - 0.0136i & 0.0994 & 3.6426 (3.4758) \\ $r$ & -0.0025 - 0.0039i &
0.1066 & 0.2793 (0.2734) \\ \hline%
\end{tabular}
\end{center}
\caption{Long time averaged statistics of the one-memory level system in
\eqref{multilevel}.} \label{tab1}
\end{table}%
}

\clearpage \begin{figure}[ptb]
\centering
\includegraphics[width=.51\textwidth]{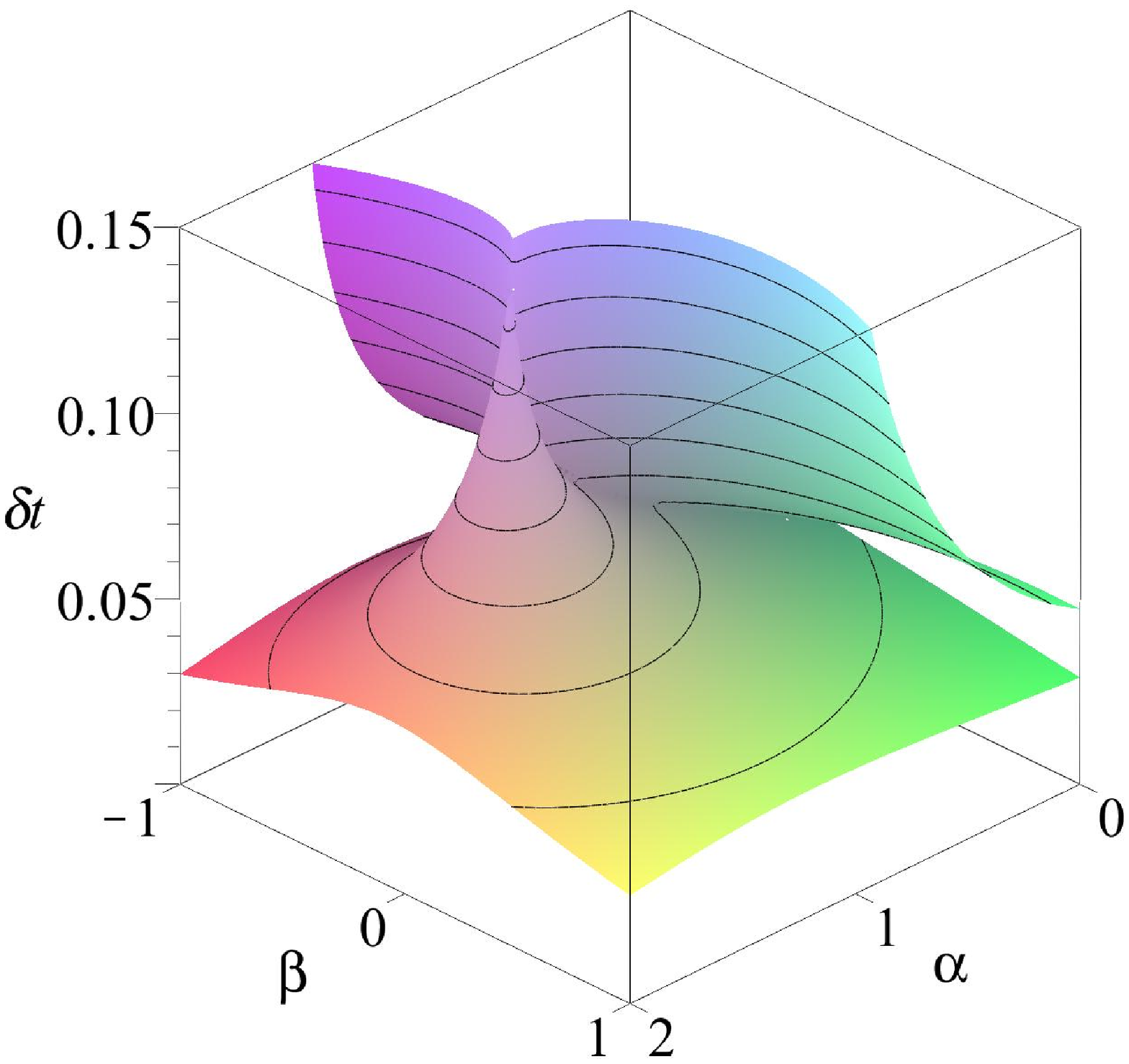}
\includegraphics[width=.44\textwidth]{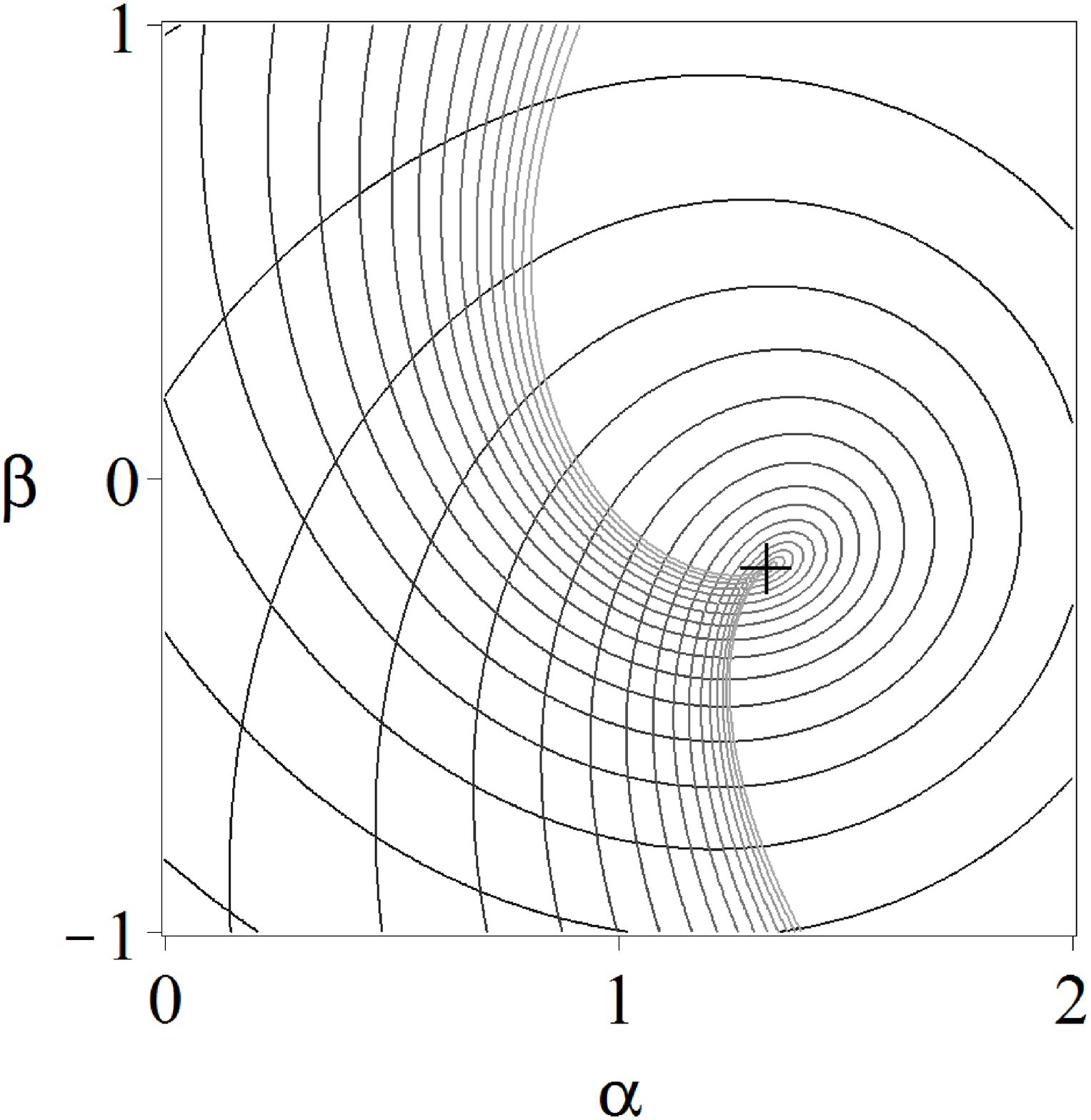}\caption{The left figure shows the surface on the three dimensional real space for $\lambda= -8.312 - 8.569i$. The
right figure shows the contour plot of the same surface on the $(\alpha,\beta)$
plane, where each curve (contour)\ corresponds to a particular value of
$\delta t.$} 
\label{fig_r3dr2d}%
\end{figure}

\clearpage \begin{figure}[ptb]
\centering
\includegraphics[width=.6\textwidth]{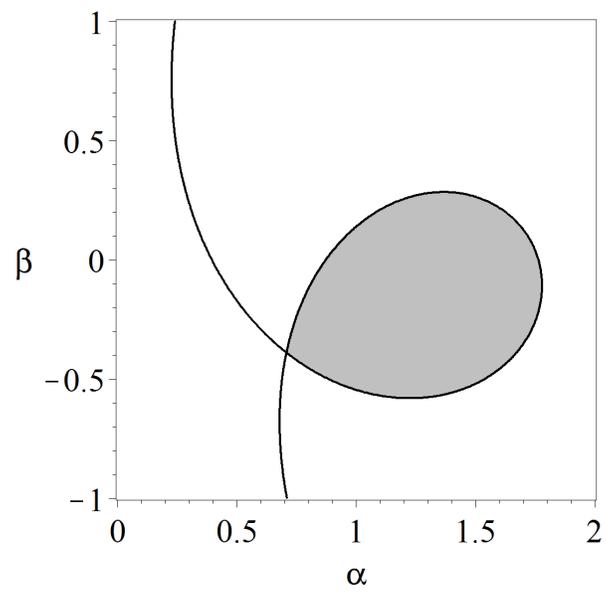}\caption{This figure shows a single contour of the surface on the $(\alpha,\beta)$ plane, for a fixed value of $\delta t$.  The stable and consistent set is shown in gray.}%
\label{fig_stable}%
\end{figure}

\clearpage \begin{figure}[ptb]
\centering
\begin{tabular}
[c]{ccccc}%
$\delta t = 0.0 \hat{\delta t}$ & $\delta t = 0.2 \hat{\delta t}$ & $\delta t = 0.4 \hat{\delta t}$ \\
\includegraphics[width=.3\textwidth]{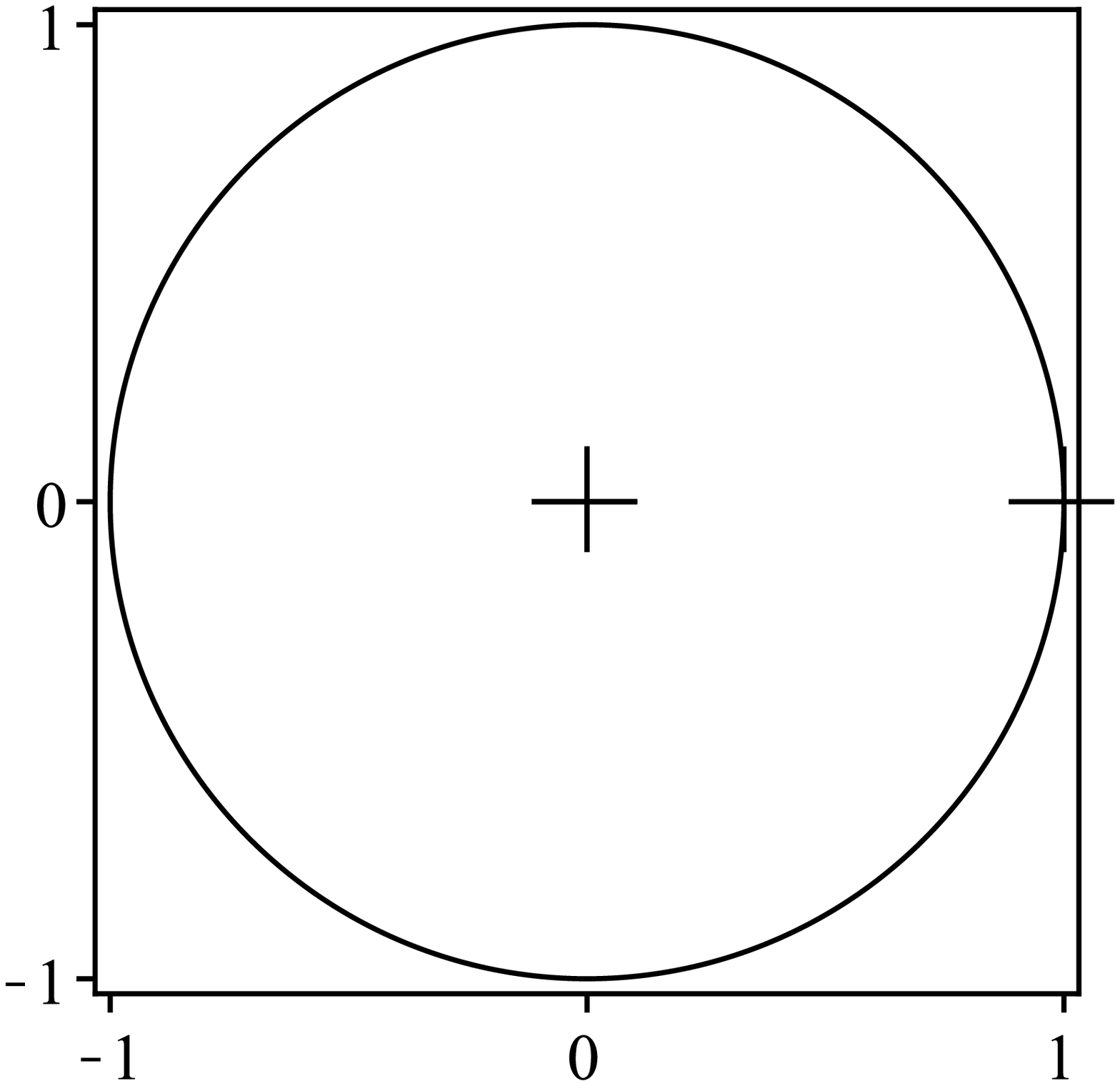} &
\includegraphics[width=.3\textwidth]{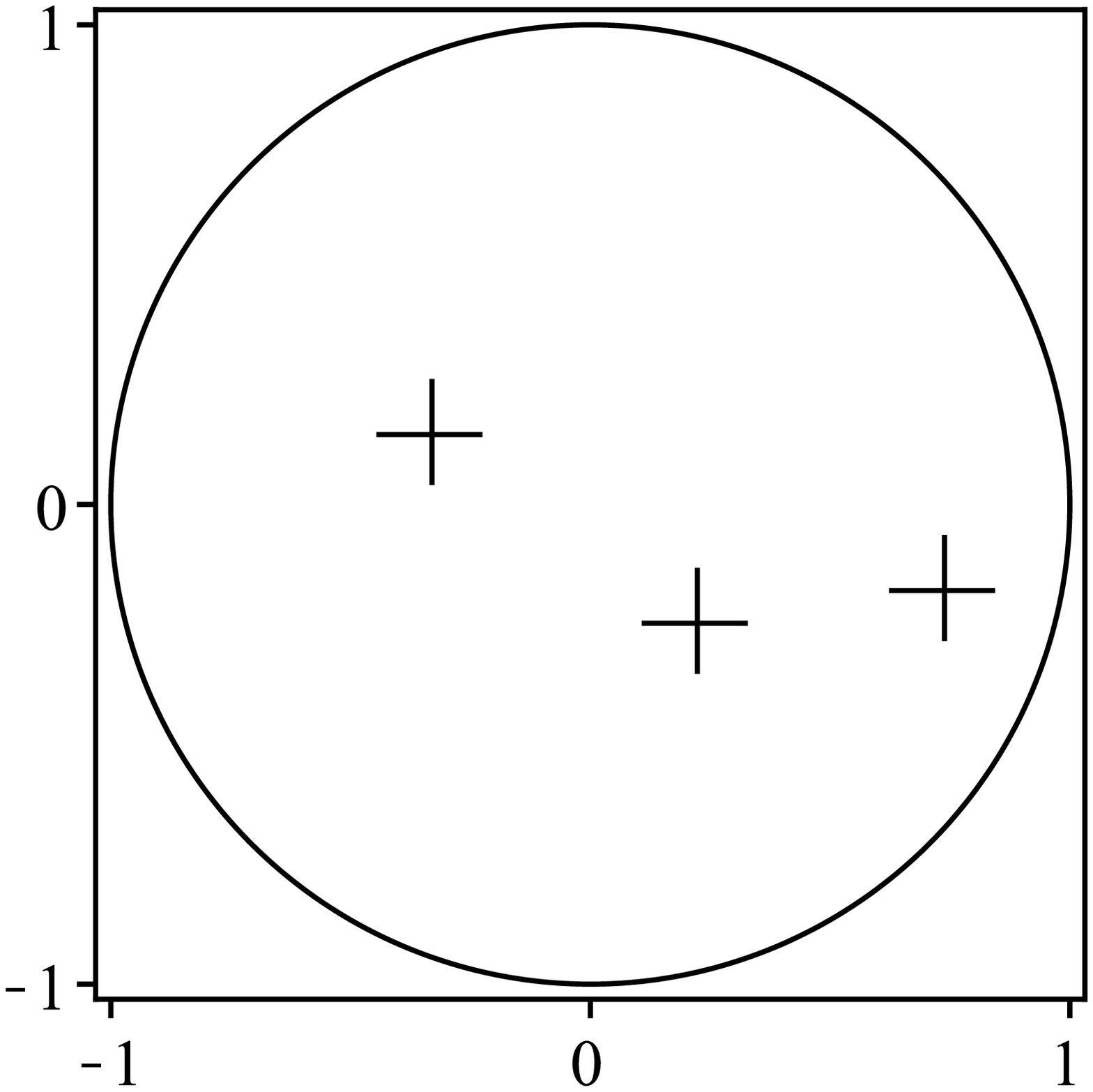} &
\includegraphics[width=.3\textwidth]{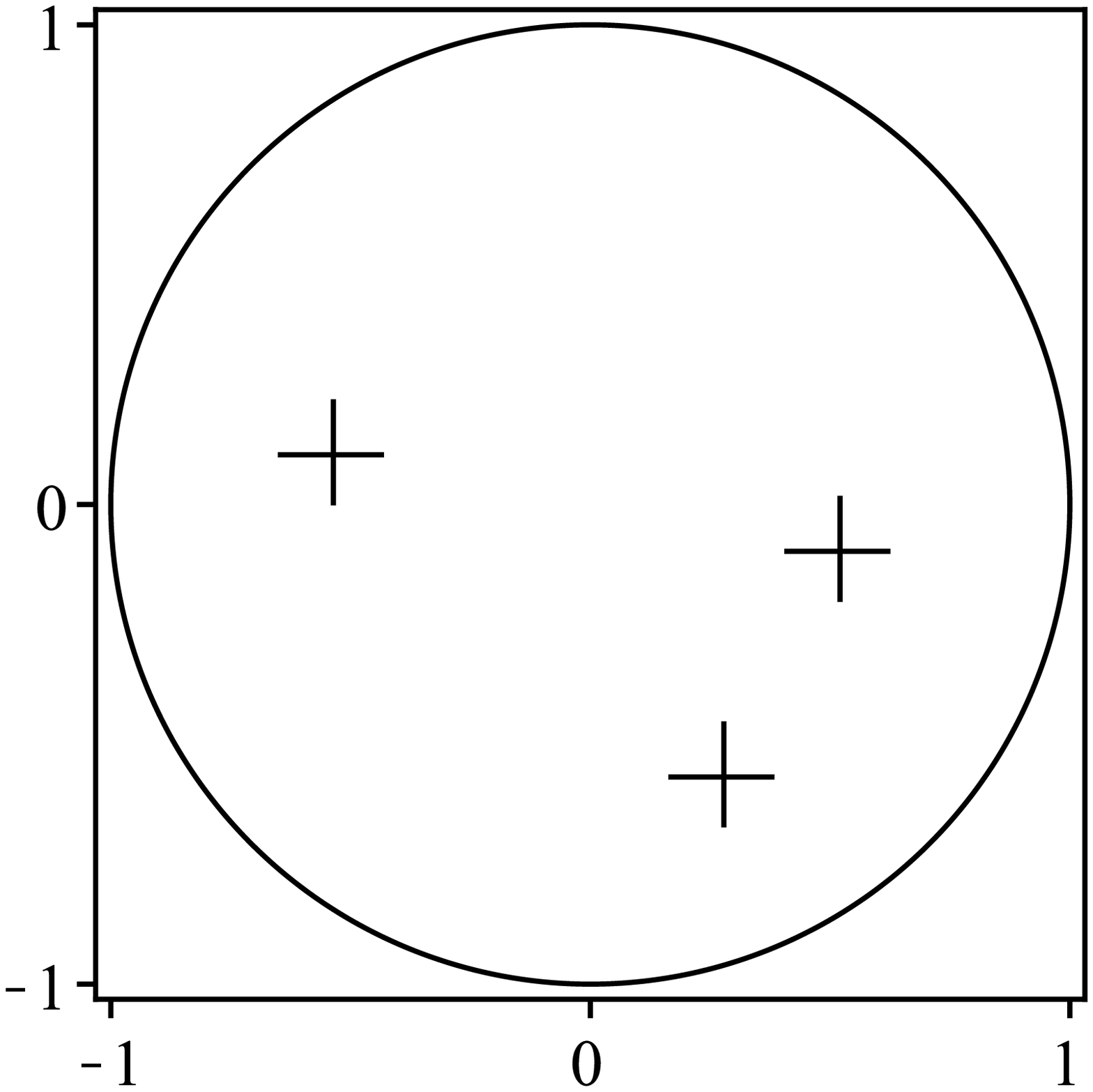} \\ \\
$\delta t = 0.6 \hat{\delta t}$ & $\delta t = 0.8 \hat{\delta t}$ & $\delta t = 1.0 \hat{\delta t}$\\
\includegraphics[width=.3\textwidth]{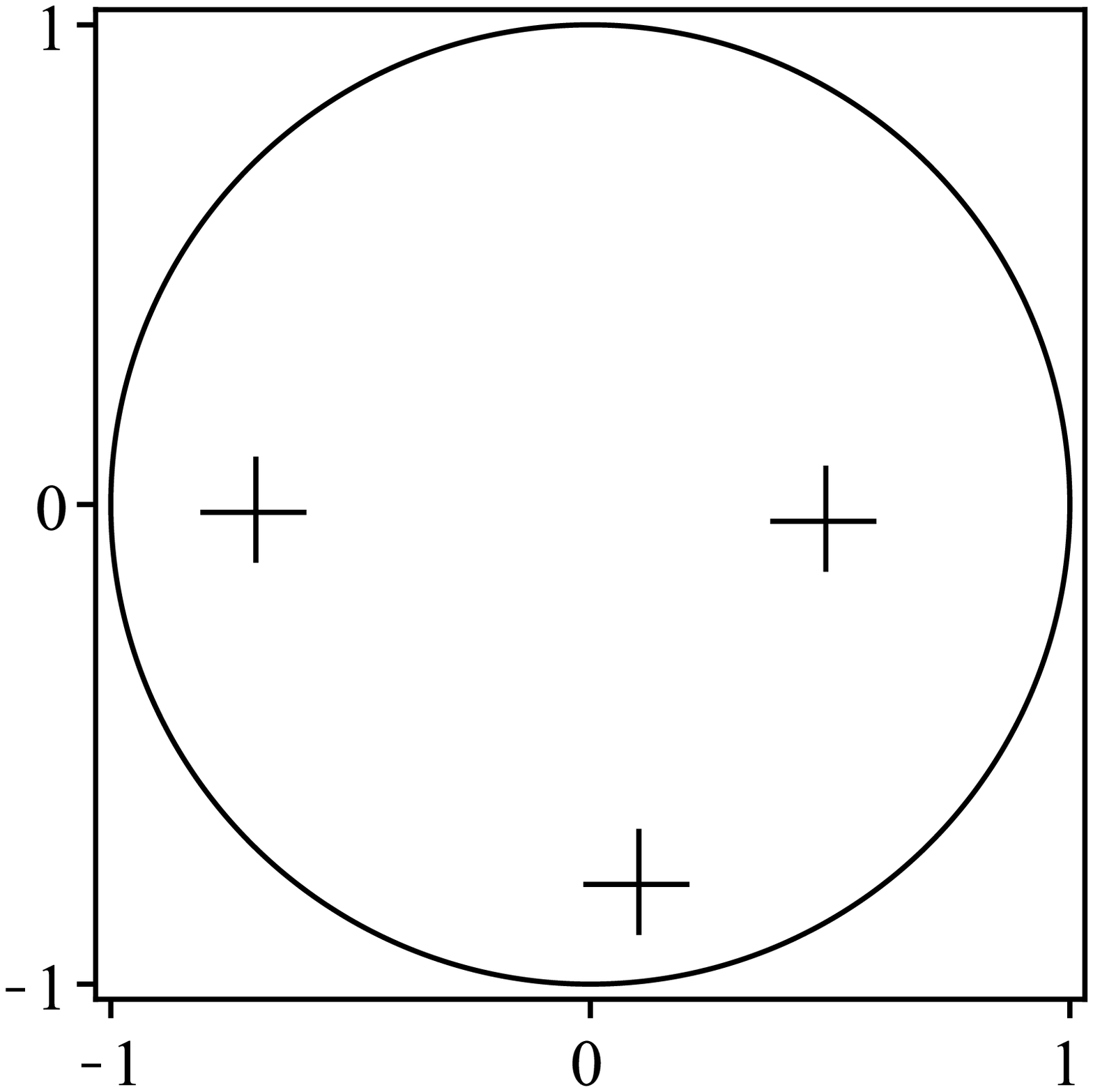} &
\includegraphics[width=.3\textwidth]{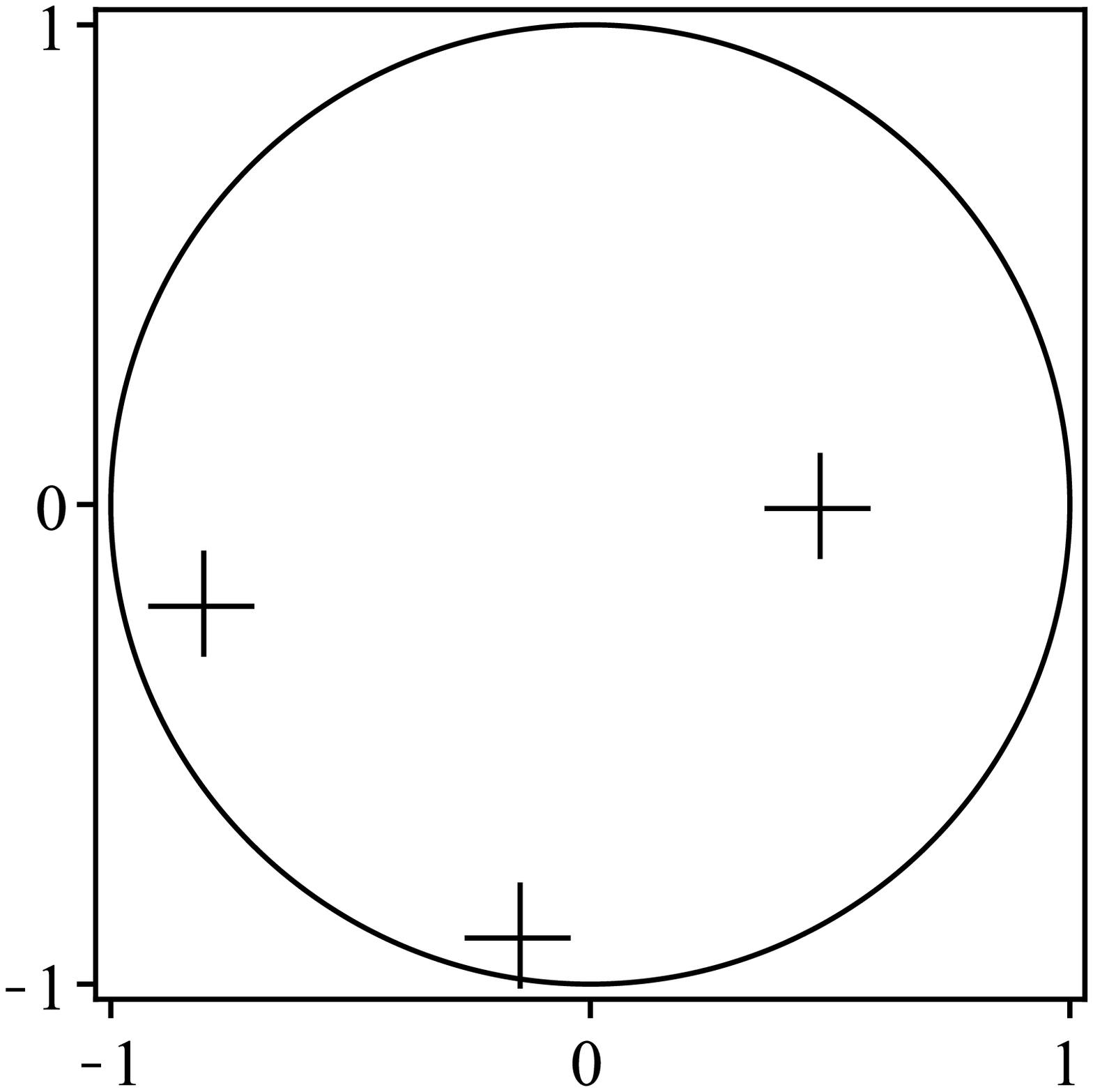} &
\includegraphics[width=.3\textwidth]{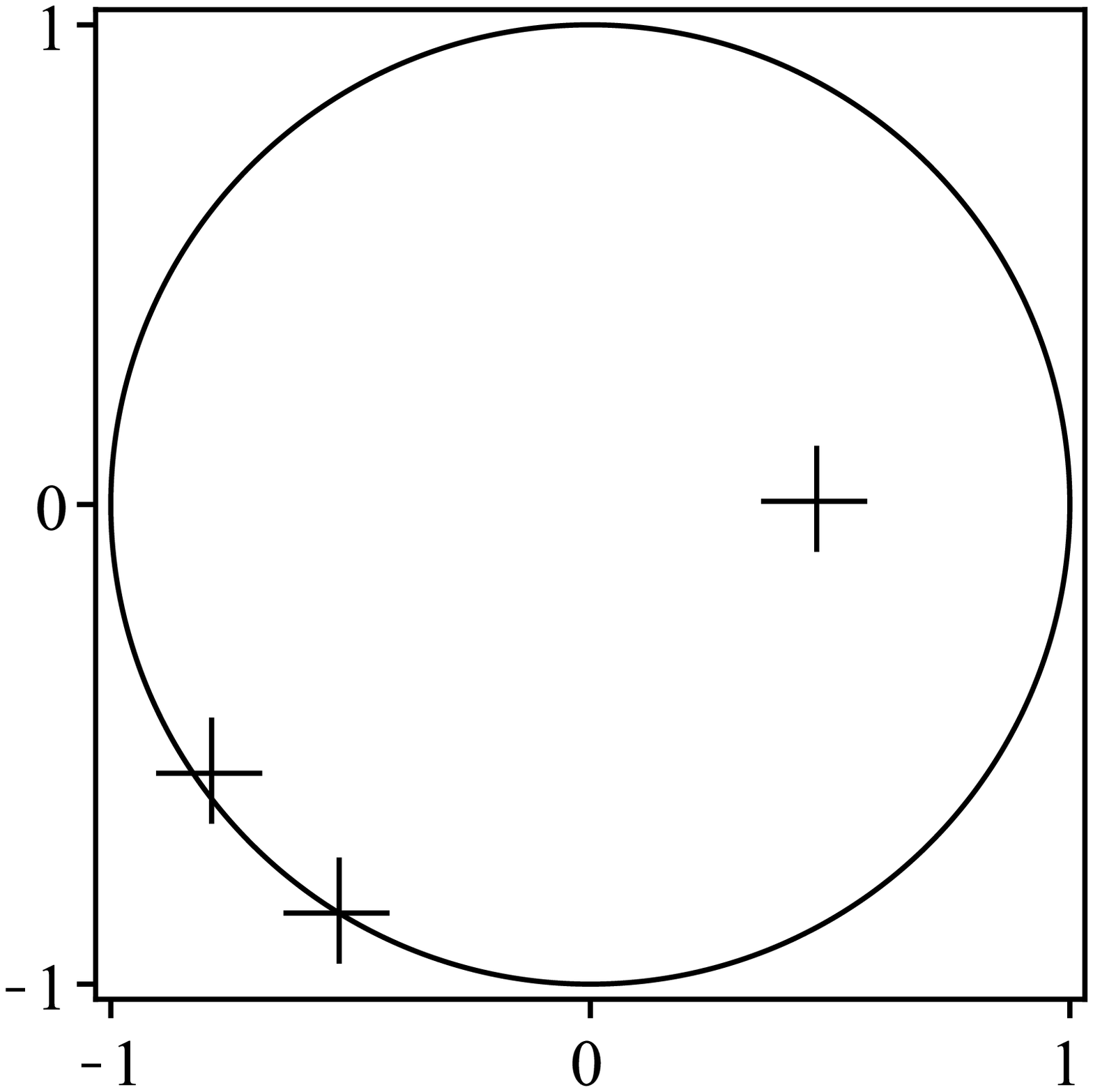}
\end{tabular}
\caption{The complex roots of $\Pi(x)$ for various $\delta t\in[0,\hat{\delta t}]$ are displayed along with the unit circle.  All roots lie strictly within the unit circle when $\delta t$ is chosen from $(0,\hat{\delta t})$.}
\label{fig_eig}
\end{figure}

\clearpage \begin{figure}[ptb]
\centering
\includegraphics[width=.6\textwidth]{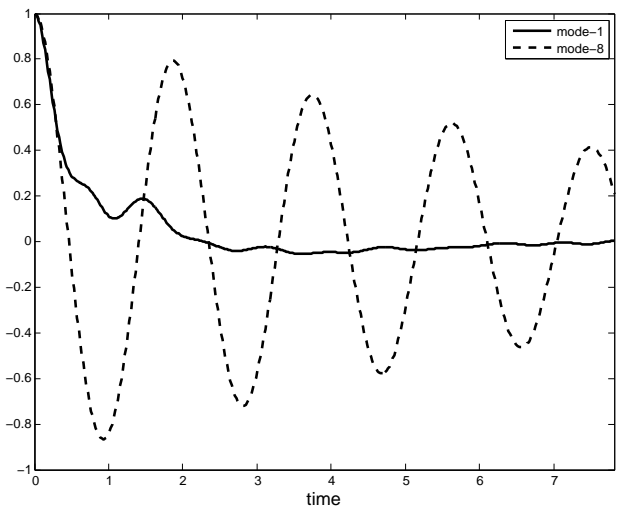}\caption{Autocorrelation
functions of Fourier modes-1 and 8.}%
\label{fig3_0}%
\end{figure}

\clearpage \begin{figure}[ptb]
\centering
\includegraphics[width=1\textwidth]{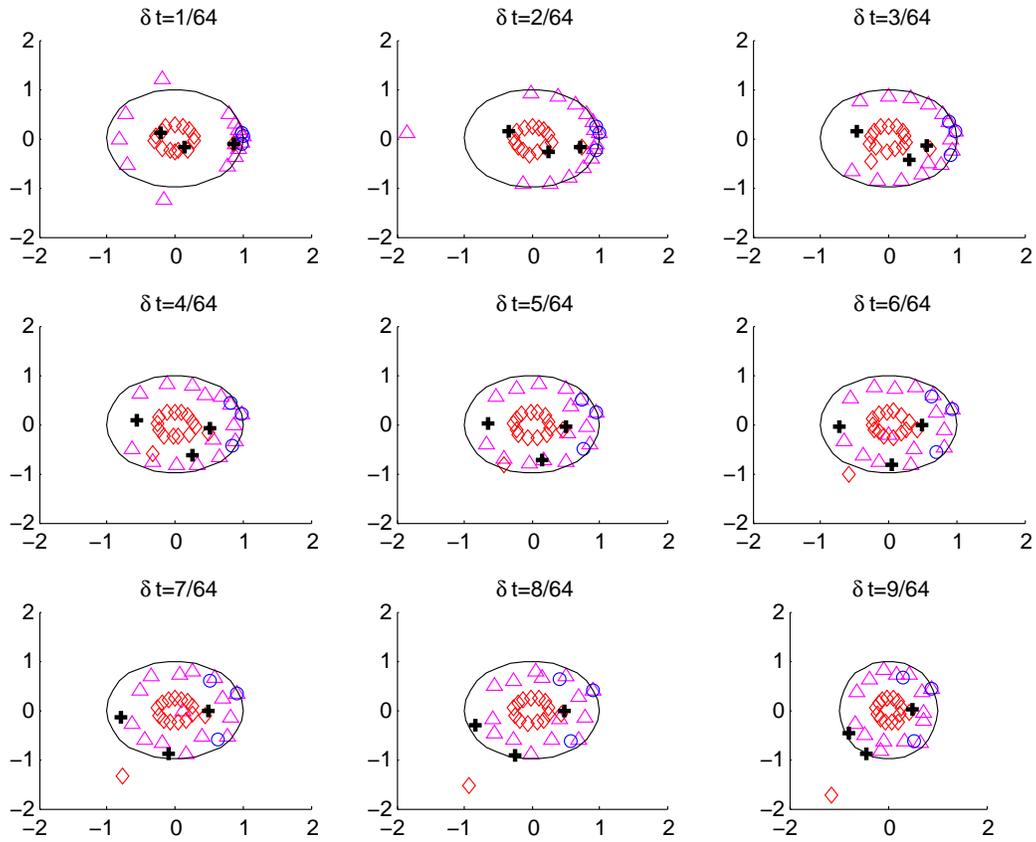}\caption{Mode-8:
Distribution of eigenvalues of AR models for various $\delta t$: AR-p model (magenta triangles); CAR-p model (red diamonds); AR-3 model (blue circles); SCAR-3 model (black plus sign).}%
\label{fig3_1}%
\end{figure}

\clearpage \begin{figure}[ptb]
\centering
\includegraphics[width=0.9\textwidth]{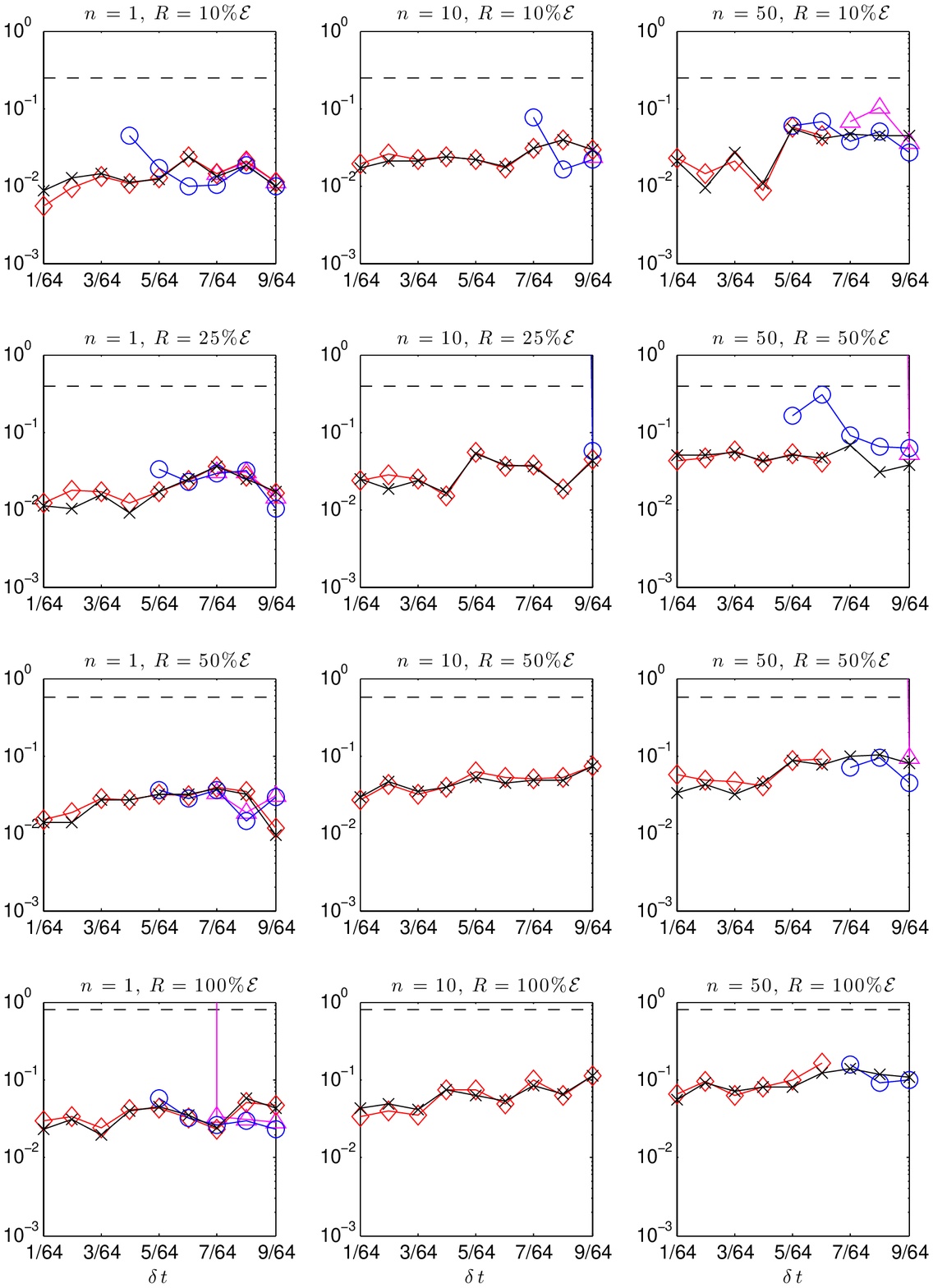}\caption{Mode-8:
Average RMSE for posterior mean estimates, $E(\hat u^{+})$, as functions of $\delta t$ for various observation
times, $\Delta t=n\delta t$, and noise variances, $R$: AR-p model (magenta
triangles); CAR-p model (red diamonds); AR-3 model (blue circles);
SCAR-3 model (black crosses); observation error (dashes).}%
\label{fig3_3}%
\end{figure}

\clearpage \begin{figure}[ptb]
\centering
\includegraphics[width=0.9\textwidth]{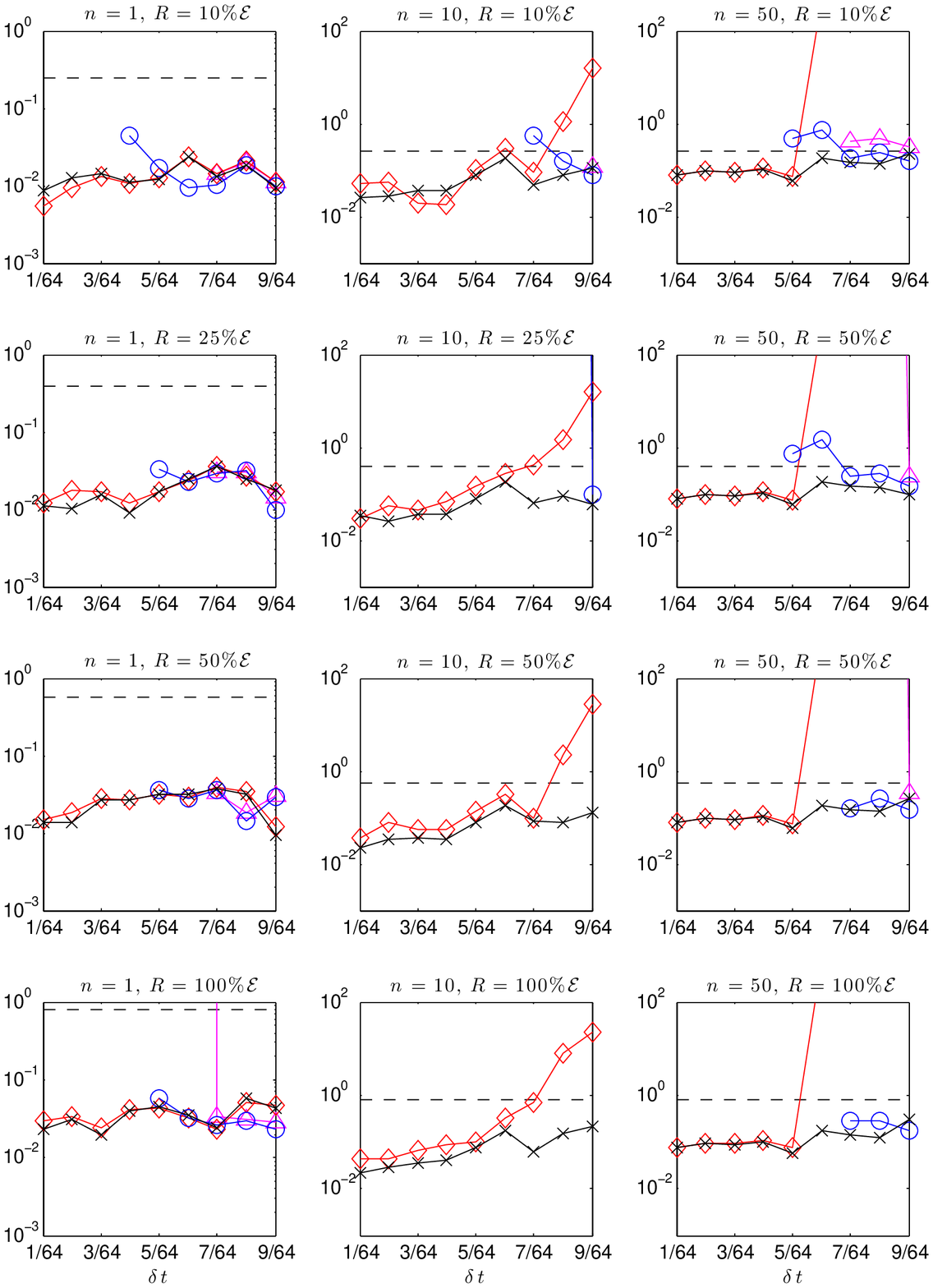}\caption{Mode-8:
Average RMSE for prior mean estimates, $E(\hat u^{-})$, as functions of $\delta t$ for various observation 
times, $\Delta t=n\delta t$, and noise variances, $R$: AR-p model (magenta
triangles); CAR-p model (red diamonds); AR-3 model (blue circles);
SCAR-3 model (black crosses); observation error (dashes).}%
\label{fig3_4}%
\end{figure}

\clearpage \begin{figure}[ptb]
\centering
\includegraphics[width=0.9\textwidth]{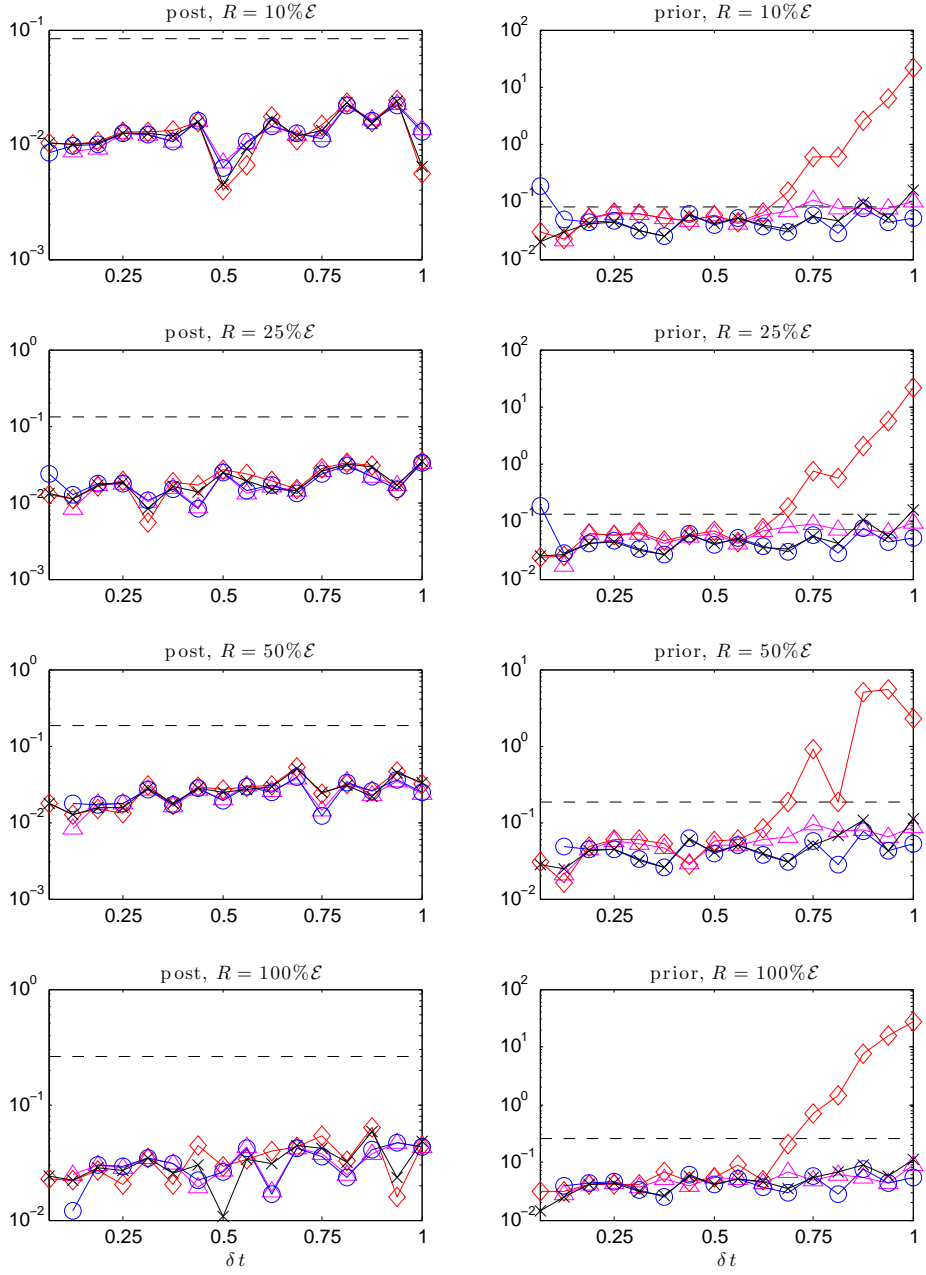}\caption{Mode-1:
Average RMSE for the posterior (first column) and prior (second column) mean estimates as functions of integration
time, $\delta t$, for observation time, $\Delta t=10\delta t$ and various
observation noise variances $R$: AR-p model (magenta
triangles); CAR-p model (red diamonds); AR-3 model (blue circles);
SCAR-3 model (black crosses); observation error (dashes).}%
\label{fig3_6}%
\end{figure}

\clearpage \begin{figure}[ptb]
\centering
\includegraphics[width=0.5\textwidth]{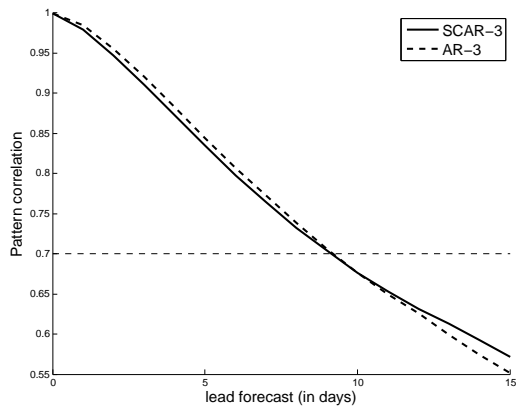}\caption{RMM index forecasting skill: Bivariate pattern correlations as functions of lead forecast (in days); obtained from daily average over a period beyond the training data set, Sept 1, 2011-Aug 31, 2012.}%
\label{pc_rmm}%
\end{figure}

\clearpage \begin{figure}[ptb]
\centering
\includegraphics[width=0.45\textwidth]{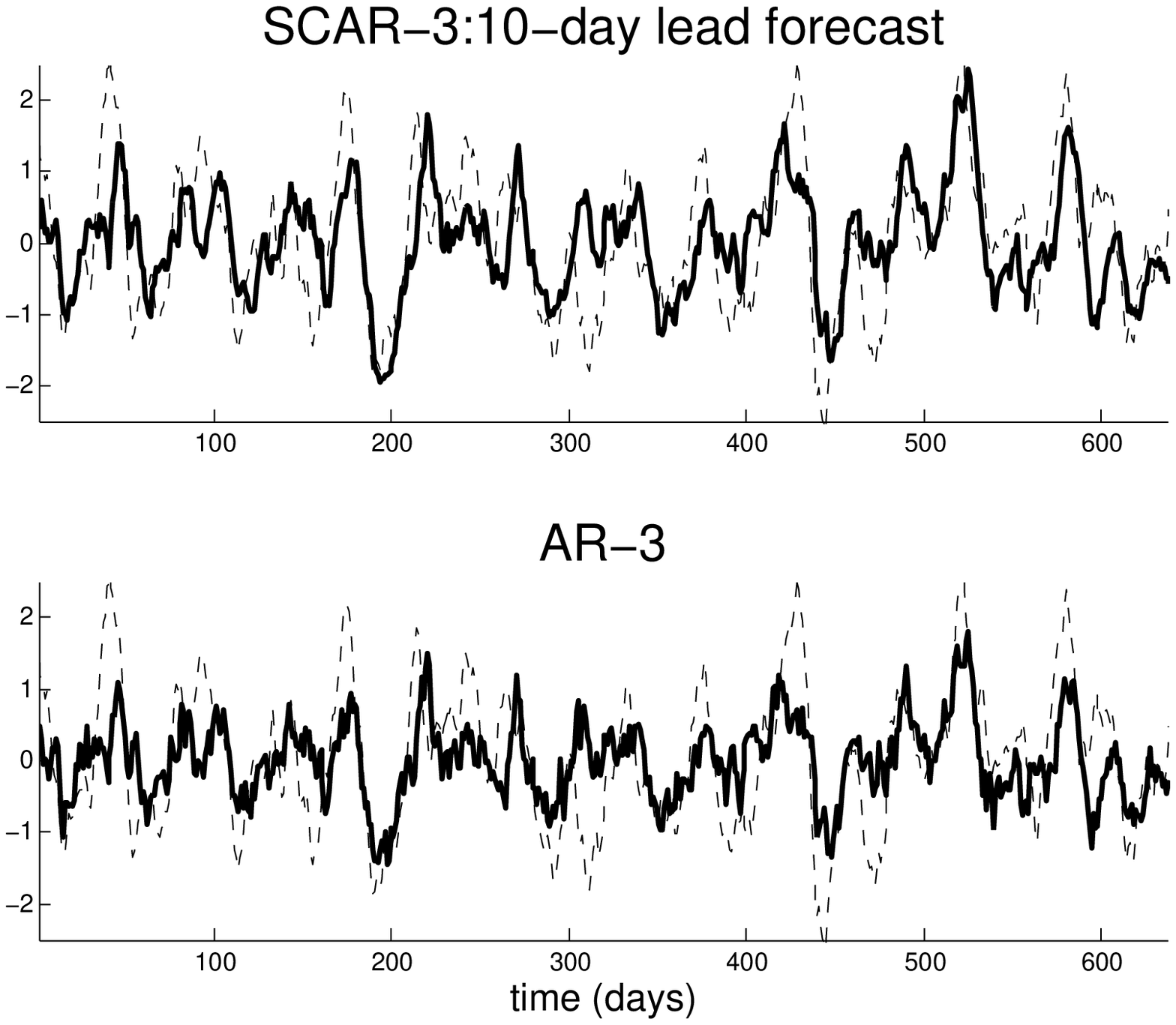}
\includegraphics[width=0.45\textwidth]{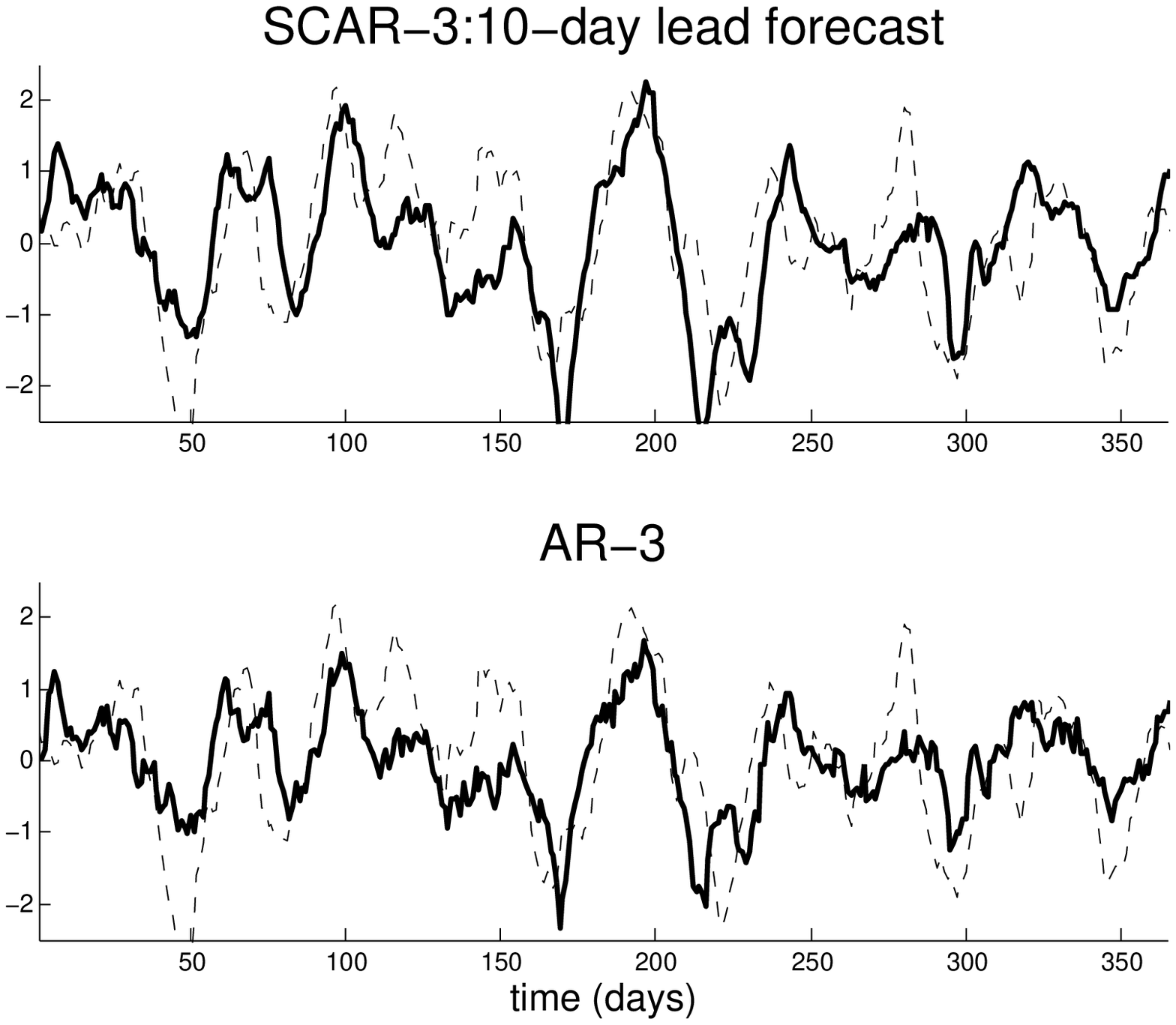}
\caption{Ten-day lead forecast for RMM1 index for periods within the training data set, Aug 1, 2005-April 30, 2007 (left), and beyond the training data set, Sept 1, 2011-Aug 31, 2012. In each panel, the observed RMM index data is denoted in dashes and the mean forecast is denoted in solid.}%
\label{fcst}%
\end{figure}


\end{document}